\documentclass[10pt,a4paper,reqno]{amsart}
\usepackage{amsthm,amssymb,amsmath}
\usepackage{mathrsfs,setspace,pstricks,multicol,multirow, latexsym}
\usepackage{epsfig, subcaption, graphics,graphicx}
\usepackage{dsfont,float}
\usepackage{enumerate}
\usepackage[per-mode=symbol,detect-weight]{siunitx}
\usepackage{a4wide}
\usepackage{appendix}
\usepackage{hyperref}
\usepackage{xcolor}
\usepackage{epstopdf}

\newcolumntype{P}[1]{>{\centering\arraybackslash}p{#1}}
\newcommand{\abs}[1]{ \left\lvert#1\right\rvert}
\newcommand{\norm}[1]{\left\lVert#1\right\rVert}
\usepackage{mathtools}

\DeclareUnicodeCharacter{2212}{-}

\newtheorem{theorem}{Theorem}[section]
\newtheorem{proposition}[theorem]{Proposition}

\newtheorem{lemma}[theorem]{Lemma}

\newtheorem{remark}[theorem]{Remark}

\begin{document}


\author{Anindya Goswami}
\address{IISER Pune, India}
\email{anindya@iiserpune.ac.in}
\thanks{* Corresponding author}

\author{Kuldip Singh Patel*}
\address{IIT Patna, India}
\email{kspatel@iitp.ac.in}

\title[Estimation of domain truncation error for PDEs]{Estimation of domain truncation error for a system of PDEs arising in option pricing}
\thanks{Authors acknowledge the support for the financial support from Department of Science $\&$ Technology under the grant no. DST/INT/DAAD/P-12/2020 and from National Board for Higher Mathematics under the grant number 02011-32-2023-R$\&$D-II-13347.}

\begin{abstract}
In this paper, a multidimensional system of parabolic partial differential equations arising in European option pricing under a regime-switching market model is studied in details. For solving that numerically, one must truncate the domain and impose an artificial boundary data. By deriving an estimate of the domain truncation error at all the points in the truncated domain, we extend some results in the literature those deal with option pricing equation under constant regime case only. We differ from the existing approach to obtain the error estimate that is sharper in certain region of the domain. Hence, the minimum of proposed and existing gives a strictly sharper estimate. A comprehensive comparison with the existing literature is carried out by considering some numerical examples. Those examples confirm that the improvement in the error estimates is significant.
\end{abstract}
\maketitle

{\bf Keywords:} regime switching market model, existence and uniqueness of solution, theory of system of PDEs, far field boundary error estimates, near field error estimates

\section{Introduction}\label{sec:intro}
\par A multidimensional system of parabolic partial differential equations (PDEs), arising in European option pricing in a regime-switching market model, is considered in this paper. Its one-dimensional version appears in \cite{ GoGG11, DESG, DKR} and in many other related works. The system of PDEs under consideration is given by
\begin{align}
\left(\frac{\partial \phi}{\partial t}+\mathbb{L}\phi\right)(t,s,i)&=0,\label{eq:main_ivp}\\
\phi(T,s,i)=\mathcal{K}(s), \quad \label{eq:main_ivp_1}
\end{align}
for all $t\in (0,T)$, $s \in (0,\infty)^d$, and $i\in \mathcal{X} :=\{1,2,...,k\}$ where
\begin{align}\label{eq:operator}
\mathbb{L}\phi(t,s,i) :=& \left(r(i) \sum_{l=1}^{d}s_l \frac{\partial \phi}{\partial s_l}+\frac{1}{2}\sum_{l=1}^{d}\sum_{l'=1}^{d}a_{ll'}(i)s_ls_{l'}\frac{\partial^2 \phi}{\partial s_l \partial s_{l'}}-r(i)\phi \right)(t,s,i)+\sum_{j=1}^{k}\lambda_{ij}\phi(t,s,j).
\end{align}
Here $a(i)=(a_{ll'}(i))_{d\times d}$ is the diffusion matrix defined as $a_{ll'}(i) := \sum_{j=1}^{d} \sigma_{l,j}(i)\sigma_{l',j}(i)$, where the $d\times d$ matrix $\sigma$, having $(l,l')$ entry $\sigma_{l,l'}: \mathcal{X}\rightarrow \mathbb{R}$, is called the volatility matrix. The matrix $\Lambda:=(\lambda_{ij})_{k\times k}$ is the instantaneous transition rate of a Markov chain on $ \mathcal{X}$, and the terminal data $\mathcal{K}(\cdot)$ is non-negative and Lipschitz continuous.
\par We recall that the existence and uniqueness results of the classical solution of a general class of system of parabolic PDEs  appear in \cite[Theorem 9.3, page 256]{Friedman83},  and \cite[Theorem 9.6, page 260]{Friedman83} respectively. It is evident that the second order operator in  \eqref{eq:main_ivp} is not strictly elliptic, which is necessary for applying the results of \cite{Friedman83}. However, after a suitable transformation (as in \cite{DurH15}), the system of PDEs (\ref{eq:main_ivp})-(\ref{eq:operator}) can be written as a constant coefficient problem. Finally, the existence and uniqueness of the resulting constant coefficient system follow from a general result \cite[Theorem 2.10]{Friedman83}. The analytical solution for the special cases of (\ref{eq:main_ivp})-(\ref{eq:operator}) have been discussed in \cite{Buffington02, Naik93, ZHU20122744}. However, all these literature are limited to two state regime-switching economy, and the analytical solutions are expressed in term of cumbersome integrals. Therefore, the development of numerical techniques is unavoidable for solving (\ref{eq:main_ivp})-(\ref{eq:operator}). For numerical computation of the solution, one must truncate the domain and impose an artificial boundary condition. Thus for estimating as well as reducing the truncation error, the growth analysis of the solution on the unbounded domain becomes essential. We prove the positivity and at most linear growth properties of the solution of (\ref{eq:main_ivp})-(\ref{eq:operator}) by using the growth of the terminal data. These results do not follow from \cite{Friedman83}. Moreover, as an offshoot, we provide a self-contained proof for the existence and uniqueness results for (\ref{eq:main_ivp})-(\ref{eq:operator}) via the semigroup approach. This does not require a transformation to a constant coefficient problem. To the best of our knowledge, the proposed approach is novel.
\par Various numerical methods have been developed in 
\cite{GoGG11, BoyD07, RR2022, Huang11, Khaliq12, Khaliq09} to solve the system (\ref{eq:main_ivp})-(\ref{eq:operator}). However, an analysis of truncation error is absent in these works. A detailed analysis of truncation error for a multidimensional PDE, i.e., when $\mathcal{X}$ is a singleton, appears in \cite{KaN00}. We extend that analysis in the settings of the system of PDEs in this paper. More precisely, using the derived growth estimate, we have estimated the error at the the boundary caused due to the artificial boundary data. Furthermore, an abstract error estimate is obtained at an interior point of the domain using a parameterized function. Subsequently, the point-wise estimate is expressed in terms of the model parameters and the maximum error on the far boundary. 
\par Our analysis helps to improve the error estimate given in \cite[Theorem 4]{KaN00} in three different ways. First of all, the error bound has been obtained for a system of multi-dimensional PDEs instead of a single multidimensional PDE. Secondly, the proposed error bound is sharper in certain region of the domain. Hence, the minimum of both gives a strictly sharper estimate. Finally, the proposed expression of the error estimate is valid for all the points in the domain unlike the expression given in \cite[Theorem 4]{KaN00} which works only on a strictly smaller subdomain. We have also included a comparison with the existing literature by considering some numerical examples with realistic parameter values. Those illustrations confirm that the improvement in the error estimates is significant. 
\par The present paper is structured as follows: The regime switching market dynamics is briefly presented in Sec. \ref{ssec:motivation}. The system of PDEs is studied in Sec. \ref{sec:exis_uniq}, and the existence and uniqueness of the solution is proved. The problem on truncated domain is considered in Sec. \ref{sec:far_estimate}. In this section, far boundary error estimates and near field estimates are derived. Some numerical examples are presented to verify the theoretical claims. Sec. \ref{sec:conclu} includes the concluding remarks with some future research directions. Appendices \ref{appendixA} and \ref{sec:appendixB} are supplemented to provide the proofs of Lemmas, which are used to prove the main results.
\section{Regime switching market dynamics}\label{ssec:motivation}
\par Assume that $r:\mathcal{X}\to [0,\infty)$, $\mu_l: \mathcal{X}\to \mathbb{R}$, and $\sigma^l:\mathcal{X}\to \mathbb{R}^d$ are given positive functions for each $l=1,\ldots,d$. We consider a friction-less market consisting of one locally risk-free asset with price $S_0$ and $d$ risky assets which may be referred to as stocks with prices $S_1$, \ldots, $S_d$. Consider the filtered probability space $(\Omega,\mathcal{F},\{\mathcal{F}_t\}_{t\ge 0}, P)$ and assume that $X:=\{X_t\}_{t\ge 0}$ and $W:=\{W_t\}_{t\ge 0}$ are $\mathcal{X}$ valued Markov chain and a $d$-dimensional standard Brownian motion respectively which are independent to each other and adapted to $\{\mathcal{F}_t\}_{t\ge 0}$. We further assume that the Markov chain $X$ on the statespace $\mathcal{X}$ has instantaneous transition rate $\Lambda:=(\lambda_{ij})_{k\times k}$. Let $r(X_t)$ be the floating interest rate of the ideal bank at time $t$. Therefore, $S_0$ at time $t$, given by $S_0(t)$ solves
\begin{equation*}
dS_0(t) = r(X_t) S_0(t) dt, \quad S_0(0) = 1.
\end{equation*}
Consider the following stochastic differential equation (SDE) 
\begin{align}\label{eq:sde}
	dS_l(t) & = S_l(t) \left[ \mu_l(X_{t} )dt +\sum_{j=1}^{d} \sigma_{l,j}(X_{t}) ~dW_t^j \right],   \\
    S_l(0) &= s_l, \quad s_l > 0,  \nonumber
    \end{align}
where $W^j:=\{W_t^j\}_{t\geq 0}$ is the $j^{\text{th}}$ component of $W$ for each $j=1,\ldots, d$.
Here $\mu_l$ and $\sigma^l=(\sigma_{l,1},\ldots,\sigma_{l,d})$ represent the growth rate and volatility coefficient of $l^{\text{th}}$ asset respectively. We further assume that for each $i\in \mathcal{X}$, $\sigma(i)$ is invertible. Since each coefficient is adapted, linear in space variable, and bounded in time variable, \eqref{eq:sde} has an almost sure unique strong solution. The price of the $l^{\text{th}}$ stock at time $t$ is modelled by $S_l(t)$. We denote $(S_1(t),\dots,S_d(t))$ by $S(t)$ and $\{S(t)\}_{t\ge 0}$ by $S$. Since $\sigma(i)$ is invertible for each $i\in \mathcal{X}$, $W$ is also an adapted $d$-dimensional standard Brownian motion w.r.t. the completion of filtration generated by $(S,X)$. Without loss of generality, we assume this filtration $\{\mathcal{F}_t \}_{t\ge 0}$ to be right continuous. 
We consider $\mathcal{K}(S(T))$ as the payoff at terminal time $T$. Let the locally risk-minimizing price of the above payoff in European style at time $t$ be denoted as  $\phi(t,s,i)$ where at $t$ the stock prices are $s=(s_1,s_2,...,s_d)$, and Markov chain is at state $i$. Then it has been shown in \cite{ GoGG11, DESG, DKR} that $\phi$ solves \eqref{eq:main_ivp}-\eqref{eq:operator} classically.
\section{Existence and Uniqueness}\label{sec:exis_uniq}
\subsection{Existence}\label{ssec:existence_ivp}
\par In order to show existence of classical solution to (\ref{eq:main_ivp})-(\ref{eq:operator}), we consider the following integral equation (IE)
\begin{align}\label{eq:integral_equation}
\phi(t,s,i)=e^{-\lambda_{i}(T-t)}\eta_{i}(t,s)+\int_{0}^{T-t}e^{-\{\lambda_{i}+r(i)\}v}\sum_{j\neq i}\lambda_{ij} \int_{x \in (0,\infty)^d}\phi(t+v,x,j)\alpha(x,s,i,v) dx dv,
\end{align}
where $\lambda_i=|\lambda_{ii}|$, and 
\begin{equation}\label{eq:alpha}
\left\{\begin{aligned}
\alpha(x,s,i,v) & \coloneqq \frac{\exp\left[-\frac{1}{2}\sum_{l=1}^d\sum_{l'=1}^d (\mathfrak{S}^{-1}_{ll'})(z_l-\tilde{z}_l)(z_{l'}-\tilde{z}_{l'})\right]}{\sqrt{(2\pi)^d |\mathfrak{S}|}\:x_1.x_2...x_d},\\
\mathfrak{S}(i)  = va(i), &\quad  z_l=\ln\left(\frac{x_l}{s_l}\right), \quad \tilde{z}_l =\left(r(i)-\frac{1}{2}a_{ll'}(i)\right)v.
\end{aligned}\right.
\end{equation}
Here, $\abs{\mathfrak{S}}$ denotes the determinant of $\mathfrak{S}$. We also recall that $a(i)=\sigma(i) \sigma(i)^{*}$, where transpose of the matrix $\sigma$ is denoted as $\sigma^*$. By setting $z=(z_1,z_2,...,z_d)$, $\tilde{z}=(\tilde{z}_1,\tilde{z}_2,...,\tilde{z}_d)$, and $\ln(s)=(\ln s_1, \ldots, \ln s_d)$, we introduce $Z \sim N(\ln (s) + \tilde{z}, \mathfrak{S}(i))$, a $d$-dimensional normal random variable for every $i\in \mathcal{X}, s \in (0,\infty)^d$, and $v >0$. Clearly, $\alpha$ is the density of the multi-dimensional lognormal random vector $(e^{Z_1}, e^{Z_2},...,e^{Z_d})$. Let $\mathfrak{D}$ and $C(\mathfrak{D})$ denote $(0,T)\times (0, \infty)^d \times \mathcal{X}$ and the space of all real-valued component-wise continuous functions on $\mathfrak{D}$, respectively. By setting $\|s\|_1= \sum_1^d |s_i|$, we define
\begin{equation}\label{eq:V_Space}
V:=\left\{\phi \in C(\mathfrak{D}) \: \Big | \norm{\phi}_V\ :=\sup_{s,t,i} \abs{\frac{\phi(t,s,i)}{1+\norm{s}_1}} < \infty \right\}. 
\end{equation}
\par The existence and uniqueness results for the integral equation (\ref{eq:integral_equation}) are obtained in Theorem \ref{theorem:IE_solution}, whereas the smoothness of the solution is established in Theorems \ref{theorem:delphidelt} and \ref{theorem:delphidels}. Finally, Theorem \ref{theorem:IEalsoPDE} states that the solution to the IE indeed solves the system of PDEs (\ref{eq:main_ivp})-(\ref{eq:operator}). Below some important properties of $\alpha(x,s,i,v)$ given in (\ref{eq:alpha}) are stated which are crucial for subsequent analysis, and the proofs of those are given in Appendix \ref{appendixA}.
\begin{lemma}\label{lemma:used_in_2} For every $s\in (0,\infty)^d, v >0,$ and $i\in \mathcal{X}$, we have 
\[
\int_{(0,\infty)^d}\left(1+\sum_{l=1}^d x_l\right)\alpha(x,s,i,v)dx=1+\sum_{l=1}^ds_l e^{r(i)v}.
\]
\end{lemma}
\begin{lemma}\label{lemma:delalphadelv}For a fixed $s\in (0,\infty)^d, v >0,$ and $i\in \mathcal{X}$, we have 
\[
\abs{\frac{1}{\alpha}\frac{\partial \alpha}{\partial v}}=\mathcal{O} \left(\ln\norm{x}_1\right)^2,\quad \mbox{as} \quad \norm{x}_1 \rightarrow \infty. 
\]
\end{lemma}
\begin{lemma}\label{lemma:delalphadels}
For a fixed $s\in (0,\infty)^d, v >0,$ $i\in \mathcal{X}$, and $1 \leq l_0 \leq n$, we have
\[
\frac{\partial \alpha}{\partial s_{l_0}}=\frac{\alpha}{s_{l_0}}\mathcal{O}(\ln{\norm{x}}),\quad \mbox{as} \quad \norm{x}_1 \rightarrow \infty.
\]
\end{lemma}
\par Following lemma includes a tail behaviour of the density function $\alpha$. Indeed, the present family of density functions with family parameter $v$ coming from a bounded set can be asymptotically dominated by a single density function with a larger value of the parameter $v$. 
\begin{lemma}\label{lemma:tail}
Given $s,s'\in (0,\infty)^d$, and $v', \epsilon >0$ there exists a sufficiently large $\mathcal{R}$ such that $\ln{\alpha(x,s',i,v')}>\ln{\alpha(x,s,i,v)}$ for all $x \notin (0,\mathcal{R})^d$ and for all $ v\le v'-\epsilon$.
\end{lemma}
\par For each $i\in \mathcal{X}$, we consider the following initial value problem (IVP) 
\begin{align}\label{eq:BS_equation}
\left(\frac{\partial \eta_i}{\partial t}+r(i)\sum_{l=1}^{d}s_l \frac{\partial \eta_i}{\partial s_l}+\frac{1}{2}\sum_{l=1}^{d}\sum_{l'=1}^{d}a_{ll'}(i)s_ls_{l'}\frac{\partial^2 \eta_i}{\partial s_l \partial s_{l'}}\right)(t,s)=r(i)\eta_i (t,s),
\end{align}
on $(0,T)\times (0,\infty)^d$ with terminal condition $\eta_i(T,s)=\mathcal{K}(s)$. For fixed $i$, (\ref{eq:BS_equation}) is known as the Black-Scholes-Merton PDE for European option price with payoff $\mathcal{K}(S(T))$. Due to the linear growth property of the payoff $\mathcal{K}(\cdot)$, $\eta_i(t,s)$ also has at most linear growth in $s$ variable for every $i$. Indeed, a log transformation of variables, followed by applications of \cite[Theorem 9.3, page 256]{Friedman83} and Feynman-Kac formula \cite[Theorem 4.4.2, page 268]{KarS98} give
\begin{equation}\label{eq:eta}
\eta_i(t,s)=E\left[e^{-r(i)(T-t)}\mathcal{K}(Y(T))|Y(t)=s\right],
\end{equation}
where $Y$ satisfies the SDE: $dY_l(t) = Y_l(t) \left[ r(i)dt +\sum_{j=1}^{d} \sigma_{l,j}(i) ~dW_t^j \right].$ Hence, $x \mapsto \alpha(x,y,i,t)$ gives the density of $Y(t)$ for each $t$, having $Y(0)=y$. Consequently, at most linear growth property of $\eta_i(t,\cdot)$ follows from the direct application of Lemma \ref{lemma:used_in_2}. Moreover, from \eqref{eq:eta} the non-negativity of $\eta$ is evident as $\mathcal{K}\ge 0$.
\begin{remark}\label{remark1}
Throughout this article, the various constant notations appear in several subsequent proofs. They carry the same meaning inside a single proof but may have different meanings when they appear in a different proof.
\end{remark}
\begin{theorem}\label{theorem:IE_solution}
Let 
\begin{align}\label{eq:operatorA}
\mathcal{A}\phi(t,s,i)&:= e^{-\lambda_{i}(T-t)}\eta_{i}(t,s)+\int_{0}^{T-t}e^{-\{\lambda_{i}+r(i)\}v}\sum_{j\neq i}\lambda_{ij} \int_{(0,\infty)^d}\phi(t+v,x,j)\alpha(x,s,i,v) dx dv,
\end{align}
for every $\phi \in V$. Then for every $\phi \in V$,
\begin{enumerate}[(i)]
    \item $\mathcal{A}\phi\in C(\mathfrak{D})$,
    \item $\mathcal{A}\phi\in V$ and $\mathcal{A}:V\to V$ is a contraction,
    \item IE (\ref{eq:integral_equation}) has a unique solution in the Banach space $(V,\norm{\cdot}_{V})$.
\end{enumerate}
\end{theorem}
\begin{proof}
$(i)$ To show that $\mathcal{A}\phi$ is in $V$, its continuity in $t$ and $s$ variables will be shown first. As the first term on the right side of (\ref{eq:operatorA}) is continuous, it is enough to show the continuity of the integral term with respect to $t$ and $s$ variables.    

\noindent\textbf{Continuity in $t$ variable:}
We need to show that 
\begin{equation*}
\begin{split}
&\left[\int_{0}^{T-t-\epsilon}e^{-(\lambda_i+r(i))v}\sum_{j \neq i}\lambda_{ij}\int_{(0,\infty)^d}\phi(t+v+\epsilon,x,j)\alpha(x,s,i,v) dxdv\right.\\ 
&\left. -\int_{0}^{T-t}e^{-(\lambda_i+r(i))v}\sum_{j \neq i}\lambda_{ij}\int_{(0,\infty)^d}\phi(t+v,x,j)\alpha(x,s,i,v) dxdv \right] \rightarrow 0\: \text{as} \: \epsilon \rightarrow 0.\nonumber
\end{split}
\end{equation*}
Note that if $\epsilon>0$ we can split the second integral in two parts, where first part will be an integration from $0$ to $\epsilon$, and second part will be an integration from $\epsilon$ to $T-t$. On the other hand, if $\epsilon<0$, we can split the first integral in two parts where the first part will be an integration from $0$ to $T-t$ and second part will be an integration from $T-t$ to $T-t-\epsilon$. Since the analyses for $\epsilon>0$ and $\epsilon<0$ are very similar, we present only for the $\epsilon>0$ case. Using a suitable substitution of variables in  the first integral term, we get
\begin{align*}
&\int_{\epsilon}^{T-t}\sum_{j \neq i}\lambda_{ij}\int_{(0,\infty)^d}\phi(t+v,x,j)\left(e^{-(\lambda_i+r(i))(v-\epsilon)}\alpha(x,s,i,v-\epsilon)-e^{-(\lambda_i+r(i))v}\right.\\
&\left.\alpha(x,s,i,v)\right)dxdv-\int_{0}^{\epsilon}e^{-(\lambda_i+r(i))v}\sum_{j \neq i}\lambda_{ij}\int_{(0,\infty)^d}\phi(t+v,x,j)\alpha(x,s,i,v) dxdv,\\
&=term1+term2 \: (say).\nonumber
\end{align*}
As $\phi$ is in $V$, the direct use of Lemma \ref{lemma:used_in_2} implies that the integral with respect to $x$ variable appearing in $term2$ is a bounded function in $v$. Hence $term2$ is an integral of a bounded function in $v$ over the interval $(0,\epsilon)$. Therefore as $\epsilon \rightarrow 0$, $term2$ $\rightarrow 0$.
\par We denote the factor $\left[e^{-(\lambda_i+r(i))(v-\epsilon)}\alpha(x,s,i,v-\epsilon)-e^{-(\lambda_i+r(i))v}\alpha(x,s,i,v)\right]$ of integrand of $term1$ as $\beta_{\epsilon}(x,s,i,v)$. Since $\alpha$ and $e^{-(\lambda_i+r(i))v}$ are continuously differentiable, we can write using mean value theorem (MVT):
\[
\beta_{\epsilon}(x,s,i,v)=-\epsilon \frac{\partial }{\partial v}\left[e^{-(\lambda_i+r(i))(v-\epsilon_1)}\alpha(x,s,i,v-\epsilon_1)\right],
\]
for some $0< \epsilon_1 < \epsilon$. From Lemma \ref{lemma:delalphadelv}, there exist $C_1(v)$ and $C_2(v)$ such that
\begin{equation}\label{eq:lemma3}
\frac{\partial \alpha}{\partial v}\leq \left(C_1+C_2\ln^2\norm{x}_1\right)\alpha \:\: \forall \:\: x \notin B,
 \end{equation}
where $B=(0,\mathcal{R})^d$, for a sufficiently large $\mathcal{R}$. For each $v>0$, inner integral of $term1$ is equal to
\begin{align}\label{eq:theorem2_two}
\epsilon\int_{(0,\infty)^d}\phi(t+v,x,j)e^{-(\lambda_i+r(i))(v-\epsilon_1)}
\Big((\lambda_i+r(i))\alpha(x,s,i,v-\epsilon_1)-\frac{\partial \alpha}{\partial v}(x,s,i,v-\epsilon_1)\Big)&dx.
\end{align}
We know that if $ 0 < \epsilon \ll v$, $\sup_{0 < \epsilon_1 < \epsilon}\alpha(x,s,i,v-\epsilon_1)$ is bounded on $B$. We write the integral (\ref{eq:theorem2_two}) as sum of two integrals by decomposing the domain $(0,\infty)^d$ as union of $B$ and $B^c$. For the integral over $B$ with a finite domain and uniformly bounded integrand, the convergence is obvious due to the dominated convergence theorem (DCT). Next from (\ref{eq:lemma3}), we note that the integrand over $B^c$ is dominated by $(C_3+C_4\norm{x}_1)(C_5+C_2 \ln^2\norm{x}_1)\alpha(x,s,i,v-\epsilon_1)$ for some $C_2, C_3, C_4$, and $C_5$ chosen independent of $\epsilon_1$. On the other hand, using Lemma \ref{lemma:tail} we get some $\mathcal{R}>0$ such that 
\[
\sup_{0<\epsilon_1 < \epsilon}\alpha(x,s,i,v-\epsilon_1)\leq \alpha(x,s,i,v+2\epsilon) \: \forall \: x \notin B.
\]
Using above inequality it is evident that integrand in (\ref{eq:theorem2_two}) is dominated by $(C_6+C_7\norm{x}_1^2) \alpha( x, s, i, v+2\epsilon)$ on $B^c$. Now we have the following claim:
\begin{equation}\label{eq:exit_claim1}
\int (C_6+C_7\norm{x}_1^2)\alpha(x,s,i,v+2\epsilon)dx \xrightarrow[]{\text{converges to}} \int (C_6+C_7\norm{x}_1^2)\alpha(x,s,i,v)dx < \infty.
\end{equation}
To prove above claim, we note that the integrand of L.H.S. of (\ref{eq:exit_claim1}) is a product of a fixed quadratic function, and a lognormal density. This is uniformly integrable family of functions in $x$ where family parameter, $v+2\epsilon$, vary on a bounded set away from $\{0\}$. This family is also tight as a consequence of tightness of Gaussian measures with parameters from a bounded set (bounded mean and variance). Then from generalized Vitali's theorem ($pp. \:\:98$, \cite{Royden10}) (\ref{eq:exit_claim1}) holds.
Using (\ref{eq:exit_claim1}), and General Lebesgue Convergence Theorem (Theorem $19$, Chapter $4$ in \cite{Royden10}), we assert that as $\epsilon\rightarrow 0$, the integral in (\ref{eq:theorem2_two}) converges to
\begin{align}
\label{philimit}    
\int_{(0,\infty)^d}\phi(t+v,x,j)\left[-\frac{\partial }{\partial v}\left(e^{-(\lambda_i+r(i))v}\alpha(x,s,i,v)\right)\right] dx.
\end{align}
However, the expression in (\ref{eq:theorem2_two}) is product of $\epsilon$ and the integral. Hence, for the outer integral w.r.t. $v$ variable in $term1$, the integrand converges to zero pointwise as $\epsilon$ goes to $0$. For a fixed $s$, this convergence is indeed uniform as the integrand of R.H.S. of (\ref{eq:exit_claim1}) is a bounded function of $v$ over the interval $(0,T-t)$. Thus using the DCT, $term1$ converges to $0$ as $\epsilon \rightarrow 0$.

\noindent \textbf{Continuity in $s_l$ variables:} Let $1_l$ denote the unit vector along $l$th direction and $1_l(l')$ denote the $l'$th component of $1_l$, and $$\gamma_{\epsilon}(x,s,i,v):=\left(\alpha(x,s+\epsilon 1_l,i,v)-\alpha(x,s,i,v)\right).$$ We need to show that for each $1 \leq l \leq d$
\begin{equation}\label{eq:conttinuityins}
\int_{0}^{T-t}e^{-(\lambda_i+r(i))v}\sum_{j \neq i}\lambda_{ij}\int_{(0,\infty)^d}\phi(t+v,x,j)\gamma_{\epsilon}(x,s,i,v)dxdv \rightarrow 0,
\end{equation}
as $\epsilon \rightarrow 0$ from either sides. Since $\alpha$ is continuously differentiable, we write using the MVT,
$\gamma_{\epsilon}(x,s,i,v)=\epsilon \frac{\partial \alpha}{\partial s_l}(x,s+\epsilon_1 1_l,i,v)$, for some $0< \abs{\epsilon_1} < \abs{\epsilon}$. Now, inner integral of (\ref{eq:conttinuityins}) is equal to
\begin{equation}\label{eq:theorem2_twos}
\epsilon\int_{(0,\infty)^d}\phi(t+v,x,j)\frac{\partial \alpha}{\partial s_l}(x,s+\epsilon_1 1_l,i,v)dx.
\end{equation}
We know that $\sup_{0 < \abs{\epsilon_1} < \abs{\epsilon}}\phi(t+v,x,j)\frac{\partial \alpha}{\partial s_l}(x,s+\epsilon_1 1_l,i,v)$ is bounded on $B$ for each $s$, $i$, $v$, $l$, $j$, and $\abs{\epsilon}< s_l$. We write the integral (\ref{eq:theorem2_twos}) as sum of two integrals by decomposing the domain $(0,\infty)^d$ as union of $B$ and $B^c$. For the first integral with a finite domain and bounded integrand, the convergence is obvious due to the DCT. Hence, we consider the integral on $B^c$ only. Again, Lemma \ref{lemma:delalphadels} guarantees the existence of constants $C_1(v)$ and $C_2(v)$ such that 
\[
\frac{\partial \alpha}{\partial s_l}\leq \left(C_1+C_2\ln\norm{x}_1\right)\frac{\alpha}{s_{l}}, \:\: \forall \:\: x \notin B,
 \]
where $B=(0,\mathcal{R})^d$ for some $\mathcal{R}>0$. Consequently, on $B^c$, the integrand of (\ref{eq:theorem2_twos}) is dominated by $\frac{1}{s_l-\abs{\epsilon}}(C_3+C_4\norm{x}_1)(C_1+C_2 \ln{\norm{x}_1})\alpha(x,s+\epsilon_1 1_l,i,v)$ for some $C_1, C_2, C_3$ and $C_4$ chosen independent of $\epsilon_1$. On the other hand using Lemma \ref{lemma:tail} we get
\[
\sup_{\abs{\epsilon_1} < \abs{\epsilon}}\alpha(x,s+\epsilon_11_l,i,v)\leq \alpha(x,s,i,v+\abs{\epsilon}), \: \forall \: x \notin B.
\]
Using this, it is evident that expression in (\ref{eq:theorem2_twos}) is dominated by $$\frac{1}{s_l-\abs{\epsilon}}(C_6+C_5\norm{x}^2_1)\alpha(x,s,i,v+\abs{\epsilon}).$$ Since (\ref{eq:exit_claim1}) holds, we also have
\begin{equation}\label{eq:exit_claim1s}
\int (C_6+C_5\norm{x}^2_1)\alpha(x,s,i,v+\abs{\epsilon})dx \xrightarrow[]{\text{converges to}} \int (C_6+C_5\norm{x}^2_1)\alpha(x,s,i,v)dx < \infty.
\end{equation}
Next, the convergence of outer integral of (\ref{eq:conttinuityins}) can be argued in a similar way as done for the outer integral of $term1$ while proving continuity in $t$ variable. Hence $\mathcal{A}\phi \in C(\mathfrak{D})$.

\noindent $(ii)$ Prior to show that the range of $\mathcal{A}$ is in $(V,\norm{\cdot}_{V})$, we consider
\begin{align*}
&\norm{\mathcal{A}\phi_1-\mathcal{A}\phi_2}_V = \sup_{D} \abs{\frac{\mathcal{A}\phi_1-\mathcal{A}\phi_2}{1+\norm{s}_1}},
\end{align*}
\begin{align*}
= &\sup_{D} \abs{\int_{0}^{T-t}e^{-(\lambda_{i}+r(i))v}\sum_{j \neq i}\lambda_{ij}\int_{(0,\infty)^{n}}\frac{\abs{\phi_1-\phi_2}(x)\alpha(x,s,i,v)}{(1+\norm{s}_1)}dx dv}, \nonumber\\
= &\sup_{D} \abs{\int_{0}^{T-t}e^{-(\lambda_{i}+r(i))v}\sum_{j \neq i}\lambda_{ij}\int_{(0,\infty)^{n}}(1+\norm{x}_1)\frac{\abs{\phi_1-\phi_2}(x)}{(1+\norm{x}_1)}\frac{\alpha(x,s,i,v)}{(1+\norm{s}_1)}dx dv}, \nonumber\\
\leq & \norm{\phi_1-\phi_2}_V \sup_{D}\abs{\int_{0}^{T-t}e^{-(r(i)+\lambda(i))v}\sum_{j\neq i}\lambda_{ij}\int_{(0,\infty)^d}(1+\norm{x}_1)\frac{\alpha(x,s,i,v)}{(1+\norm{s}_1)}dx dv}.
\end{align*}
We simplify the above inequality using the result in Lemma~\ref{lemma:used_in_2}, and obtain
\begin{align}\label{eq:contraction}
\norm{\mathcal{A}\phi_1-\mathcal{A}\phi_2}_V & \leq \norm{\phi_1-\phi_2}_V \sup_{D}\abs{\int_{0}^{T-t}e^{-(r(i)+\lambda(i))v}\sum_{j\neq i}\lambda_{ij}\frac{1+\sum_{l=1}^{d}s_le^{r(i)v}}{1+\sum_{l=1}^ds_l }dv},\nonumber\\
&\leq \norm{\phi_1-\phi_2}_V\sup_{D}\int_{0}^{T-t}e^{-\lambda_i v}\sum_{j\neq i}\lambda_{ij}dv,\nonumber\\
&= \norm{\phi_1-\phi_2}_V\sup_{D}\int_{0}^{T-t}\lambda_i e^{-\lambda_iv}dv= C \norm{\phi_1-\phi_2}_V,
\end{align}
for some $C<1$. To show that $\norm{\mathcal{A}\phi}_V$ is finite for every $\phi \in V$, we take $\phi_2=0$ in above inequality (\ref{eq:contraction}), and get
\begin{equation*}
\begin{split}
\norm{\mathcal{A}\phi_1}_{V} &\leq \norm{\mathcal{A}\phi_1-\mathcal{A}\textbf{0}}_{V}+\norm{\mathcal{A}\textbf{0}}_{V} \leq C\norm{\phi_1}_{V}+\norm{\mathcal{A}\textbf{0}}_{V}.
\end{split}
\end{equation*}
Using (\ref{eq:operatorA}), $\mathcal{A}\textbf{0}=e^{-\lambda_{i}(T-t)}\eta_{i}(t,s)$. Hence, from the argument below Eq. \eqref{eq:eta} $\norm{\mathcal{A}\textbf{0}}_{V} < \infty$ which implies $\norm{\mathcal{A}\phi_1}_{V}$ is also finite. Again since $\mathcal{A}\phi \in C(\mathfrak{D})$, we have $\mathcal{A}:V\to V$. Thus (\ref{eq:contraction}) implies that $\mathcal{A}$ is a contraction on $V$. 

\noindent $(iii)$ A direct application of Banach fixed point theorem \cite[Theorem A1, page 528]{Limaye96} gives that $\mathcal{A}$ has a unique fixed point in $V$. Hence, 
 (\ref{eq:integral_equation}) has a unique solution.
\end{proof}

\par For showing that the unique solution $\phi\in V$ of (\ref{eq:integral_equation}) solves (\ref{eq:main_ivp})-(\ref{eq:operator}) classically, we prove in Theorems \ref{theorem:delphidelt} and \ref{theorem:delphidels} that $\phi$ has required smoothness. The following lemma is required in the proof of Theorem \ref{theorem:delphidelt} and its proof is given in Appendix \ref{appendixA}.

\begin{lemma}\label{Lemma:theorem2firstpart}
Let $\phi$ be the solution of integral equation (\ref{eq:integral_equation}). Then for each $t, s, j, j'$, we have
\begin{equation}\label{eq:lemma5_1}
\lim_{u\downarrow 0}\int_{(0,\infty)^d} \phi(t+u,x,j)\alpha(x,s,j',u)\,dx=\phi(t,s,j).
\end{equation}
\end{lemma}
\begin{theorem}\label{theorem:delphidelt} Let $\phi \in V$ and solves (\ref{eq:integral_equation}), then it is differentiable in $t$ variable. Furthermore, for $g_1(x,s,i,v) := \left(\frac{\partial \alpha}{\partial v}/\alpha\right)(x,s,i,v)$
\begin{align}\label{eq:delphi_delt}
\frac{\partial \phi}{\partial t}(t,s,i) &= -r(i)e^{-\lambda_{i}(T-t)}\eta_{i}(t,s)+e^{-\lambda_i(T-t)}\frac{\partial \eta_i}{\partial t}-\sum_{j \neq i}\lambda_{ij}\phi(t,s,j)+(\lambda_i+r(i))\times \phi(t,s,i) \nonumber\\
-&\int_{0}^{T-t}\sum_{j\neq i}\lambda_{ij}\int_{(0,\infty)^d}\phi(t+v,x,j)e^{-(\lambda_i+r(i))v} 
\times g_1(x,s,i,v)\alpha(x,s,i,v) dx dv.
\end{align}
\end{theorem}
\begin{proof} The partial derivative $\frac{\partial \phi}{\partial t}(t,s,i)$, if exists, can be written as follows:
\begin{align*}
\frac{\partial \phi}{\partial t}(t,s,i)&=\frac{\partial}{\partial t}\left(e^{-\lambda_i(T-t)}\eta_i(t,s)\right)+\lim_{\epsilon \rightarrow 0}\frac{1}{\epsilon}\left[\int_{0}^{T-t-\epsilon}e^{-(\lambda_i+r(i))v}\sum_{j \neq i}\lambda_{ij}\times\int_{(0,\infty)^d}\phi(t+v+\epsilon,x,j) \right.~~\\
&\alpha(x,s,i,v) dxdv-\int_{0}^{T-t}e^{-(\lambda_i+r(i))v}\sum_{j \neq i}\lambda_{ij} \left.\int_{(0,\infty)^d}\phi(t+v,x,j) \alpha(x,s,i,v) dxdv \right].
\end{align*}
The partial derivative, i.e. the first term on R.H.S., exists as it is the derivative of product of two smooth functions. Next, we consider the second (limit) term. Using a suitable substitution, the limit term is equals to
\begin{align*}
&\frac{1}{\epsilon}\Bigg[\int_{\epsilon}^{T-t}\sum_{j \neq i}\lambda_{ij}\int_{(0,\infty)^d}\phi(t+v,x,j)\Bigg(e^{-(\lambda_i+r(i))(v-\epsilon)}\alpha(x,s,i,v-\epsilon)-e^{-(\lambda_i+r(i))v}\\
&\alpha(x,s,i,v)\Bigg)dxdv \Bigg] -\frac{1}{\epsilon}\int_{0}^{\epsilon}e^{-(\lambda_i+r(i))v}\sum_{j \neq i}\lambda_{ij}\int_{(0,\infty)^d}\phi(t+v,x,j)\alpha(x,s,i,v) dxdv,\\
&=term1+term2 \: (say).\nonumber
\end{align*}
From Lemma \ref{Lemma:theorem2firstpart}, as $\epsilon \rightarrow 0$, $term2$ $\rightarrow -\sum_{j \neq i}\lambda_{ij}\phi(t,s,j)$. As explained in the proof of Theorem \ref{theorem:IE_solution}, the inner integral of $term1$ converges to (\ref{philimit}) with an $\epsilon$ multiplied. After cancelling the $\epsilon$ with $\frac{1}{\epsilon}$, and using the uniform boudedness of the integrand of the outer integral, we get
\[
term1 \to \int_{0}^{T-t}\sum_{j \neq i}\lambda_{ij}\int_{(0,\infty)^d}\phi(t+v,x,j)\left[-\frac{\partial }{\partial v}\left(e^{-(\lambda_i+r(i))v}\alpha(x,s,i,v)\right)\right] dxdv.
\]

\noindent Therefore, we can write
\begin{align}\label{eq:limitinside}
\frac{\partial \phi}{\partial t}&(t,s,i)
-\frac{\partial}{\partial t}\left(e^{-\lambda_i(T-t)}\eta_i(t,s)\right) + \sum_{j \neq i}\lambda_{ij}\phi(t,s,j)\nonumber\\
=&\int_{0}^{T-t}\sum_{j \neq i}\lambda_{ij}\int_{(0,\infty)^d}\phi(t+v,x,j)\left[-\frac{\partial }{\partial v}\left(e^{-(\lambda_i+r(i))v}\alpha(x,s,i,v)\right)\right] dxdv,\nonumber\\
=&\int_{0}^{T-t}\sum_{j \neq i}\lambda_{ij}\int_{(0,\infty)^d}\phi(t+v,x,j)\left[(\lambda_i+r(i))e^{-(\lambda_i+r(i))v}\alpha(x,s,i,v)-e^{-(\lambda_i+r(i))v}\frac{\partial \alpha}{\partial v}\right] dxdv,\nonumber\\
=&-\int_{0}^{T-t}\sum_{j \neq i}\lambda_{ij}\int_{(0,\infty)^d}\phi(t+v,x,j)e^{-(\lambda_i+r(i))v}\frac{\partial \alpha}{\partial v} dxdv\nonumber\\
&+(\lambda_i+r(i))\int_{0}^{T-t}e^{-(\lambda_i+r(i))v}\sum_{j \neq i}\lambda_{ij}\int_{(0,\infty)^d}\phi(t+v,x,j)\alpha(x,s,i,v)dxdv.
\end{align}
\noindent By rewriting the last term of (\ref{eq:limitinside}) using (\ref{eq:integral_equation}), the above becomes
\begin{equation}
\begin{split}
&\frac{\partial \phi}{\partial t}(t,s,i)=\frac{\partial}{\partial t}\left(e^{-\lambda_i(T-t)}\eta_i(t,s)\right)-\int_{0}^{T-t}\sum_{j \neq i}\lambda_{ij}\int_{(0,\infty)^d}\phi(t+v,x,j)e^{-(\lambda_i+r(i))v} \\
& \frac{\partial \alpha}{\partial v}(x,s,i,v)dxdv -\sum_{j \neq i}\lambda_{ij}\phi(t,s,j)+(\lambda_i+r(i))\left(\phi(t,x,i)-e^{-\lambda_i(T-t)}\eta_i\right).
\nonumber
\end{split}
\end{equation}
The simplification of right side of above expression gives (\ref{eq:delphi_delt}).
\end{proof}
\begin{theorem}\label{theorem:delphidels} If  $\phi\in V$ solves (\ref{eq:integral_equation}), then for each $l\le d$, $\phi$ is twice differentiable in $s_l$ variable. Furthermore, 
\begin{align}
\frac{\partial \phi}{\partial s_l}(t,s,i)=&e^{-\lambda_{i}(T-t)}\frac{\partial \eta_{i}}{\partial s_l}+\int_{0}^{T-t}e^{-(\lambda_i+r(i))v}\sum_{j\neq i}\lambda_{ij}\int_{(0,\infty)^d}\phi(t+v,x,j) g_2^l(x,s,i,v)\nonumber\\&\alpha(x,s,i,v) dx dv,\label{eq:delphi_dels}\\
\frac{\partial^2 \phi}{\partial s_l \partial s_{l'}}(t,s,i)=&e^{-\lambda_{i}(T-t)}\frac{\partial^2 \eta_{i}}{\partial s_l \partial s_{l'}}+\int_{0}^{T-t}e^{-(\lambda_i+r(i))v}\sum_{j\neq i}\lambda_{ij}\nonumber\\&\int_{(0,\infty)^d}\phi(t+v,x,j)\left(\frac{\partial }{\partial s_{l'}}g_2^l+g_2^lg_2^{l'}\right)(x,s,i,v)\alpha(x,s,i,v) dx dv,\label{eq:delphi_dels2}
\end{align}
where $g_2^l(x,s,i,v):=\left(\frac{\partial \alpha}{\partial s_l}/\alpha\right)(x,s,i,v)$.
\end{theorem}
\begin{proof} From (\ref{eq:integral_equation}), we can write
\begin{align}\label{eq:delphi_dels1_new}
\frac{\partial \phi}{\partial s_l}(t,s,i)=\frac{\partial}{\partial s_l} \Bigg[e^{-\lambda_i(T-t)}\eta_i+\int_{0}^{T-t}e^{-\{\lambda_{i}+r(i)\}v}\sum_{j\neq i}\lambda_{ij}\int_{(0,\infty)^d}\phi(t+v,x,j)\times \alpha(x,s,i,v) dx dv\Bigg],
\end{align}
provided the partial derivative exists. In other words, (\ref{eq:delphi_dels1_new}) can be rewritten as
\begin{align*}
\frac{\partial \phi}{\partial s_l}(t,s,i) = e^{-\lambda_i(T-t)}\frac{\partial \eta_i}{\partial s_l}+\lim_{\epsilon \rightarrow 0}\frac{1}{\epsilon}\int_{0}^{T-t}e^{-(\lambda_i+r(i))v}\sum_{j \neq i}\lambda_{ij}\int_{(0,\infty)^d}\phi(t+v,x,j)\\
\times (\alpha(x,s+\epsilon \textbf{1}_l,i,v)-\alpha(x,s,i,v)) dxdv,
\end{align*}
if the above limit exists. Now we consider the second (limit) term only. We have already proved in Theorem \ref{theorem:IE_solution} (continuity in $s$ variable part) that as $\epsilon \rightarrow 0$, this second (limit) term converges to 
\[
\int_{0}^{T-t}e^{-(\lambda_i+r(i))v}\sum_{j \neq i}\lambda_{ij}\int_{(0,\infty)^d}\phi(t+v,x,j)\frac{\partial \alpha}{\partial s_l}(x,s,i,v)dxdv,
\]
which gives the desired expression for $\frac{\partial \phi}{\partial s_l}(t,s,i)$ as in (\ref{eq:delphi_dels}), since $g_2^l\alpha:=\frac{\partial \alpha}{\partial s_l}$. Now for the second order partial derivative $\frac{\partial^2 \phi}{\partial s_l\partial s_l'}(t,s,i)$, we write
\begin{align}\label{eq:delphi_dels2_new}
\frac{\partial^2 \phi}{\partial s_l\partial s_{l'}}(t,s,i)=\frac{\partial^2}{\partial s_l\partial s_{l'}} \Bigg[e^{-\lambda_i(T-t)}\eta_i + \int_{0}^{T-t}e^{-\{\lambda_{i}+r(i)\}v}\sum_{j\neq i}\lambda_{ij}\int_{(0,\infty)^d}\phi(t+v,x,j)\nonumber \\
\times \alpha(x,s,i,v) dx dv\Bigg],
\end{align}
provided the partial derivative in (\ref{eq:delphi_dels2_new}) exists. One can show the existence of this in a similar line of that for the first order derivative with the only difference arising from the tail estimate of $\frac{\partial^2 \alpha}{\partial s_l \partial s_{l'}}$. However in this case, similar to Lemma \ref{lemma:delalphadels}, we have  $\frac{\partial^2 \alpha}{\partial s_{l} \partial s_{l'}}=\frac{\alpha}{s_ls_{l'}}\mathcal{O}(\ln \norm{x})^2$, for a fixed $s\in (0,\infty)^d, v >0$. Thus, we can dominate the second order partial derivative of $\alpha$ by the product of $\alpha$, and a quadratic function in $x$. Note that this dominating function has also arisen in Theorem \ref{theorem:IE_solution} (continuity in $s$ variable part). Hence, the rest of the proof of this theorem will exactly be the same as the proof of Theorem \ref{theorem:IE_solution} (continuity in $s$ variable part).
\end{proof}
\par We have proved that unique solution $\phi$ of IE (\ref{eq:integral_equation}) is sufficiently smooth. Now we prove that it also satisfies the system of PDEs (\ref{eq:main_ivp}) in the following Theorem.
\begin{theorem}\label{theorem:IEalsoPDE}
Let $\phi(t,s,i)$ be the unique solution of IE (\ref{eq:integral_equation}) then
\begin{enumerate}[(i)]
\item $\phi(t,s,i)$ also satisfies PDE (\ref{eq:main_ivp})-(\ref{eq:operator}).
\item $\phi(t,s,i)$ is non-negative and of at-most linear growth.
\end{enumerate}
\end{theorem}
\begin{proof}
Using (\ref{eq:delphi_delt}) from Theorem \ref{theorem:delphidelt}, and (\ref{eq:delphi_dels}) $\&$ (\ref{eq:delphi_dels2}) from Theorem \ref{theorem:delphidels}, we get
\begin{align*}
&\bigg(\frac{\partial \phi}{\partial t}+r(i)\sum_{l=1}^d s_l \frac{\partial \phi}{\partial s_l}+\frac{1}{2}\sum_{l=1}^d \sum_{l'=1}^d s_l s_{l'}a_{ll'}(i)\frac{\partial^2 \phi}{\partial s_l \partial s_{l'}}\bigg)(t,s,i)\\
&=e^{-\lambda_i(T-t)}\left[\frac{\partial \eta_i}{\partial t}-r(i)\eta_i+r(i)\sum_{l=1}^d s_l \frac{\partial \eta_i}{\partial s_l}+\frac{1}{2}\sum_{l=1}^d\sum_{l'=1}^d s_l s_{l'}a_{ll'}(i)\frac{\partial^2 \eta_i}{\partial s_l \partial s_l \partial s_{l'}}\right]\\
&-\sum_{j \neq i}\lambda_{ij}\phi(t,s,j)+(r(i)+\lambda_i)\phi(t,s,i)+\int_{0}^{T-t}\sum_{j \neq i} \lambda_{ij}\int_{(0,\infty)^d}\phi(t+v,x,i)e^{-(\lambda_i+r(i))v}\\
&\alpha(x,s,i,v)\times\left[-g_1+r(i)\sum_{l=1}^d s_l g_2^l+\frac{1}{2}\sum_{l=1}^{d}\sum_{l'=1}^d s_l s_{l'}a_{ll'}(i)\left(\frac{\partial}{\partial s_{l'}}g_2^l+g_2^lg_2^{l'}\right)\right]dxdv,
\end{align*}
which implies
\begin{align*}
\bigg(\frac{\partial \phi}{\partial t}+r(i) \sum_{l=1}^d s_l \frac{\partial \phi}{\partial s_l}+\frac{1}{2}\sum_{l=1}^d \sum_{l'=1}^d s_l s_{l'}a_{ll'}(i)\frac{\partial^2 \phi}{\partial s_l \partial s_{l'}} -r(i)\phi \bigg)(t,s,i)+\sum_{j=1}^k\lambda_{ij}\phi(t,s,j)\\
=\int_{0}^{T-t}\sum_{j \neq i} \lambda_{ij}\int_{(0,\infty)^d}\phi(t+v,x,i)e^{-(\lambda_i+r(i))v}\alpha(x,s,i,v)
\Bigg[-g_1+r(i)\\
\times\sum_{l=1}^d s_l g_2^l+\frac{1}{2}\sum_{l=1}^{d}\sum_{l'=1}^d s_l s_{l'}a_{ll'}(i)\left(\frac{\partial}{\partial s_{l'}}g_2^l+g_2^lg_2^{l'}\right)\Bigg](x,s,i,v)dxdv.
\end{align*}
We refer to Appendix \ref{appendixA} for proof of the following identity
\begin{align}\label{eq:lemma_appendixA}
\Bigg[-g_1+r(i)\sum_{l=1}^d s_l g_2^l+\frac{1}{2}\sum_{l=1}^{d}\sum_{l'=1}^d s_l s_{l'}a_{ll'}(i)
\left(\frac{\partial}{\partial s_{l'}}g_2^l+g_2^lg_2^{l'}\right)\Bigg](x,s,i,v)&=0.
\end{align}
Using \eqref{eq:lemma_appendixA}, the equation above reduces to
\begin{align*}
\bigg(\frac{\partial \phi}{\partial t}+r(i) \sum_{l=1}^d s_l \frac{\partial \phi}{\partial s_l}+\frac{1}{2}\sum_{l=1}^d \sum_{l'=1}^d s_l s_{l'}a_{ll'}(i)\frac{\partial^2 \phi}{\partial s_l \partial s_{l'}}\bigg)(t,s,i)
+\sum_{j=1}^k\lambda_{ij}\phi(t,s,j)-r(i)\phi(t,s,i)=0.
\end{align*}
Again, in (\ref{eq:integral_equation}), by substituting $t=T$, we get $\phi(T,s,i)= \eta_i(T,s)= \mathcal{K}(s)$. The last equality is due to (\ref{eq:BS_equation}).  Thus, unique solution $\phi(t,s,i)$ of IE (\ref{eq:integral_equation})  satisfies (\ref{eq:main_ivp_1}) along with system of PDEs (\ref{eq:main_ivp}).\\
(ii) Since, $\mathcal{K}$ is non-negative, $\eta$ in (\ref{eq:operatorA}) is non-negative too. Therefore, the left side of (\ref{eq:operatorA}) is non-negative, provided $\phi \ge 0$. Thus $\mathcal{A}: H \rightarrow H$, where $H$ is the set of all non-negative functions in $V$. Clearly, $H$ is a complete metric space with metric $d(h_1,h_2) = \| h_1-h_2\|_V$ too. Moreover, Theorem \ref{theorem:IE_solution}-(ii) implies $d(\mathcal{A}\phi_1, \mathcal{A}\phi_2) \leq C d(\phi_1,\phi_2)$ for some $0<C<1$. Finally, the result follows from \cite[Theorem $17.1$(a) ]{DeimK85}.
\end{proof}
\subsection{Uniqueness}\label{ssec:uniqueness}
In this subsection, we aim to prove that the system of PDEs (\ref{eq:main_ivp})-(\ref{eq:operator}) has unique classical solution in $(V,\norm{\cdot}_V)$ via probabilistic approach. Let $1 \leq l \leq d$ and  $\tilde{S}_l:= \{\tilde{S}_l(t)\}_{t\geq 0}$ be the strong solution of the following SDE
\begin{equation}\label{eq:sde2}
d\tilde{S}_l(t)=\tilde{S}_l(t)\left(r(X_t)dt+\hat{\sigma}_l(X_t)d\hat{W}_t^l\right), \quad \tilde{S}_l(0) >0,
\end{equation}
where $\hat{\sigma}_l(i) = \norm{\sigma^l(i)}_2$, and $\hat{W}_t^l=\frac{\sum_j\sigma_{l,j}(X_t)W_t^j}{\hat{\sigma}_l}$. Note that, $\hat{W}_t^l$ is a Brownian motion using \cite[Theorem $8.4.2$, pp. $143$]{Oks00}. We denote $(\tilde{S}_1,\ldots, \tilde{S}_d)$ by $\tilde{S}$, and the filtration generated by $(\tilde S, X)$ by $\{\tilde{\mathcal{F}}_t\}_t$. The following lemmas, whose proofs are given in Appendix \ref{appendixA}, are crucial to achieve the uniqueness result.
\begin{lemma}\label{lemma:Stmartingale}
The $l^{th}$ component $\tilde{S}_l$ of $\tilde{S}$ is a sub-martingale for each $l$.
\end{lemma}
\par Due to the above lemma, and the finite second order moment of $\tilde{S}_l$, we can apply the Doob's maximal inequality \cite[pp.132, Theorem 7.3.2]{ANSh19} on $\tilde{S}_l$ to get
\begin{align}\label{DMI}
E\left(\sup_{s \leq t}|\tilde{S}_l(s)|\right)< \infty,
\end{align}
for each $l=1,\ldots, d$, and $t\ge 0$.
\begin{lemma}\label{lemma:Ntmartingale}
If $\phi(t,s,i)$ is a classical solution of (\ref{eq:main_ivp}) with at most linear growth, then the process $N^\phi=\{N^\phi_t\}_{t \geq 0}$ given by
\begin{equation}\label{eq:poisson}
N^\phi_t=e^{-\int_{0}^{t}r(X_u)du}\phi(t,\tilde{S}(t),X_t),
\end{equation} 
is a martingale.
\end{lemma}
\par From Theorem \ref{theorem:IEalsoPDE}, we know that  (\ref{eq:main_ivp})-(\ref{eq:operator}) has at least one classical solution in $V$. Again, the above lemma implies that every such solution $\phi$ would produce a martingale $N^\phi$ (as in \eqref{eq:poisson}). Then using $N^\phi_t=E[N^\phi_T|\tilde{\mathcal{F}}_t]$, Markovity of $(\tilde{S},X)$, and the terminal condition of (\ref{eq:main_ivp})-(\ref{eq:operator}), we get
\begin{align}
\label{eq:option_expectation}
\nonumber 
e^{-\int_{0}^t r(X_u)du}\phi(t,\tilde{S}_t,X_t)= & E\left(e^{-\int_{0}^T r(X_u)du}\mathcal{K}(\tilde{S}(T))|\tilde{\mathcal{F}}_t\right),
\\
\textrm{ or,}\quad \phi(t,s,i)=&E\left(e^{-\int_{t}^T r(X_u)du}\mathcal{K}(\tilde{S}(T))|\tilde{S}(t)=s,X_t=i\right).
\end{align}
\par Since every classical solution of (\ref{eq:main_ivp})-(\ref{eq:operator}) in $V$ has the identical expression \eqref{eq:option_expectation}, all of them are identical. By this we have established the following theorem.
\begin{theorem}\label{theorem:uniqueness}
The problem (\ref{eq:main_ivp})-(\ref{eq:operator}) has a unique classical solution in $V$.
\end{theorem}
\par It is evident from the above discussion and Theorem \ref{theorem:uniqueness} that the classical solution of (\ref{eq:main_ivp})-(\ref{eq:operator}) having at most linear growth also solves the IE (\ref{eq:integral_equation}). However, above results do not indicate how to derive the IE from the PDE. For an independent interest, we produce the derivation of the IE (\ref{eq:integral_equation}) using the stochastic representation (\ref{eq:option_expectation}) of $\phi(t,s,i)$ in Appendix \ref{sec:appendixB}.
\section{Truncated domain problem}\label{sec:far_estimate}
\noindent Due to the absence of analytical solution of (\ref{eq:main_ivp})-(\ref{eq:operator}), it needs to be solved numerically by truncating the unbounded domain to a bounded one. Let $\psi:[0,T]\times \Pi_{l=1}^{d}[s_l^b, s_l^u]\times \mathcal{X}\to \mathbb{R}$ be such that in the interior
\begin{align}
\bigg(\frac{\partial \psi}{\partial t}+r(i)\sum_{l=1}^{d}s_l \frac{\partial \psi}{\partial s_l}+\frac{1}{2}\sum_{l=1}^{d}\sum_{l'=1}^{d}a_{ll'}(i)s_ls_{l'}&\frac{\partial^2 \psi}{\partial s_l \partial s_{l'}} -r(i)\psi\bigg)(t,s,i)+\sum_{j=1}^{k} \lambda_{ij}\psi(t,s,j) =0, \label{eq:main_ibvp}\\
\psi(T,s,i)&=\mathcal{K}(s) \forall \quad s \in R,\: i \in \mathcal{X},\label{eq:main_ibvp_ini}\\
\psi(t,s,i)&=h(t,s,i) \quad \forall \: (t,s,i)\in  (0,T)\times \Gamma\times \mathcal{X},\label{eq:main_ibvp_boun}
\end{align}
where $R=\Pi_{l=1}^{d}(s_l^b, s_l^u)$,
$0 \leq s_l^b<s_l^u$, $\forall$ $l$, and $\Gamma=\partial R \cap (0,\infty)^d$. For each $i$, $h(\cdot, \cdot,i)$ is set as a sufficiently smooth function on closure of $(0,T)\times R$. The existence and uniqueness of the classical solution of (\ref{eq:main_ibvp})-(\ref{eq:main_ibvp_boun}) can be borrowed from the Theorem $3.5$ on $pp.$ $291$ in \cite{EgS94}, and also in Theorem $10.1$ on $pp.$ $616$ in \cite{Lady67}. However, the application of these Theorems requires smoothness of $\Gamma$. For our case, this is achieved by a smooth approximation of $R$, which is explained below.
\par For any $\varepsilon\in (0,1)$, let $U^0_{\varepsilon}:= \{ s\in \mathbb{R}^d\mid \sum_{l=1}^d {|s_l|}^{1/\varepsilon} < 1 \}$ and the diagonal matrix $\mathcal{M}$ be such that the $l$th diagonal element is $\frac{s_l^u - s_l^b}{2}$. Then clearly $U_{\varepsilon} := \frac{1}{2}(s^b +s^u) + \mathcal{M} U^0_{\varepsilon}$ is contained in $R$, and having smooth boundary. Furthermore, $U_{\varepsilon_1} \subset U_{\varepsilon_2}$ for any $1\ge \varepsilon_1> \varepsilon_2>0$, and $$ \bigcup_{\varepsilon\in (0,1)} U_{\varepsilon} = R. $$ Hence, due to the smoothness of $h$, $\psi_{\varepsilon}$, the classical solution of (\ref{eq:main_ibvp})-(\ref{eq:main_ibvp_boun}), where $R$ is replaced by $U_{\varepsilon}$, approximates $\psi$ for sufficiently small $\varepsilon>0$. 
\subsection{Growth estimate}
Next, we derive a growth estimate of the solution of the un-truncated problem depending on the growth of the terminal data. This is useful in several aspects, including in estimating the error due to the boundary data of the truncated domain problem. A similar result appeared in \cite[Theorem $2$]{KaN00} for Black-Scholes-Merton PDE, which is extended for a system of PDEs here. It is worth noting that the present proof is entirely different from that given in \cite{KaN00}. 
\begin{theorem}\label{theorem:lineargrowth} Let $\phi(t,s,i)$ be the solution of (\ref{eq:main_ivp})-(\ref{eq:operator}). In addition to the non-negativity and Lipschitz continuity, we further assume that 
\begin{equation}\label{eq:solution_bound}
-k_3+k_4\cdot s \leq \mathcal{K}(s) \leq k_1+k_2\cdot s \quad \forall \quad s \in (0,\infty)^d,
\end{equation}
for some $k_1,k_3 \geq 0$ and vectors $k_2, k_4 \in \mathbb{R}^d$, where $``\cdot"$ represents inner product. Then
\begin{equation}\label{eq:theorem15}
-k_3 e^{-(\min_{i} r(i))(T-t)}+ k_4\cdot  s \leq \phi(t,s,i) \leq k_1e^{-(\min_{i} r(i))(T-t)}+k_2\cdot  s \quad \forall \quad (t,s,i) \in D.
\end{equation}
\end{theorem}
\begin{proof}
From (\ref{eq:option_expectation}) and (\ref{eq:solution_bound}), we can write,
\begin{align*}
E\left[e^{-\int_{t}^T r(X_u)du}(-k_3+k_4\cdot\tilde{S}(T))|\tilde{S}(t),X_t\right]\leq \phi(t,\tilde{S}(t),X_t) \\ \leq E\left[e^{-\int_{t}^T r(X_u)du}(k_1+k_2\cdot\tilde{S}(T))|\tilde{S}(t),X_t\right].
\end{align*}
Using the Markovity of $(\tilde S,X)$ w.r.t. $\{\tilde{\mathcal{F}}_t\}_t$, we have
\begin{align*}
E\left[e^{-\int_{t}^T r(X_u)du}(-k_3+k_4\cdot \tilde{S}(T))|\tilde{\mathcal{F}}_t\right] \leq \phi(t,\tilde{S}(t),X_t)\\ 
\leq E\left[e^{-\int_{t}^T r(X_u)du}(k_1+k_2\cdot \tilde{S}(T))|\tilde{\mathcal{F}}_t\right].
\end{align*}
Multiplying $e^{-\int_{0}^tr(X_u)du}$ to each term in the above inequality, we obtain
\begin{align*}
E\left[e^{-\int_{0}^T r(X_u)du}(-k_3+k_4\cdot  \tilde{S}(T))|\tilde{\mathcal{F}}_t\right]&\leq e^{-\int_{0}^tr(X_u)du}\phi(t,\tilde{S}(t),X_t) \\&\leq E\left[e^{-\int_{0}^T r(X_u)du}(k_1+k_2 \cdot  \tilde{S}(T))|\tilde{\mathcal{F}}_t\right],
\end{align*}
or
\begin{align*}
-k_3 E\left[e^{-\int_{0}^T r(X_u)du}|\tilde{\mathcal{F}}_t\right]  +k_4\cdot E\left[e^{-\int_{0}^T r(X_u)du} \tilde{S}(T)|\tilde{\mathcal{F}}_t\right]\leq e^{-\int_{0}^tr(X_u)du}\phi(t,\tilde{S}(t),X_t) \\ 
\leq k_1 E\left[e^{-\int_{0}^T r(X_u)du}|\tilde{\mathcal{F}}_t\right] +k_2\cdot E\left[e^{-\int_{0}^T r(X_u)du} \tilde{S}(T)|\tilde{\mathcal{F}}_t\right].
\end{align*}
Note that the function $\varphi_l\in V$ given by $\varphi_l(t,s,i):=s_l$ solves (\ref{eq:main_ivp})
classically for each $l$. By applying Lemma \ref{lemma:Ntmartingale}, we get that  $\{e^{-\int_{0}^tr(X_u)du}\tilde{S}^l_t\}_{t \geq 0}$ is martingale for each $l$. Hence, using $E\left[e^{-\int_{0}^T r(X_u)du} \tilde{S}(T)|\tilde{\mathcal{F}}_t\right] =e^{-\int_{0}^t r(X_u)du} \tilde{S}(t)$, the above inequality reduce to
\begin{align*}
   -k_3 E\left[e^{-\int_{0}^T r(X_u)du}|\tilde{\mathcal{F}}_t\right]  +k_4\cdot e^{-\int_{0}^t r(X_u)du} \tilde{S}(t)\leq 
e^{-\int_{0}^tr(X_u)du}\phi(t,\tilde{S}(t),X_t), \\\leq k_1 E\left[e^{-\int_{0}^T r(X_u)du}|\tilde{\mathcal{F}}_t\right]  +k_2\cdot e^{-\int_{0}^t r(X_u)du} \tilde{S}(t).
\end{align*}
Cancelling $e^{-\int_{0}^t r(X_u)du}$ from each term, we get
\begin{align*}
&-k_3 E\left[e^{-\int_{t}^T r(X_u)du}|\tilde{\mathcal{F}}_t\right] +k_4\cdot  \tilde{S}(t)\leq \phi(t,\tilde{S}(t),X_t) \\& \leq k_1 E\left[e^{-\int_{t}^T r(X_u)du}|\tilde{\mathcal{F}}_t\right] +k_2\cdot  \tilde{S}(t)-k_3e^{-(\min_{i}r(i))(T-t)}+k_4\cdot  \tilde{S}(t)\\& \leq \phi(t,\tilde{S}(t),X_t) \leq k_1e^{-(\min_{i}r(i))(T-t)}+k_2\cdot  \tilde{S}(t),
\end{align*}
almost surely for all $t\in [0,T]$.
Hence (\ref{eq:theorem15}), obtained by replacing $\tilde{S}(t)=s$ and $X_t=i$, follows for almost every $(t,s,i) \in D$, since $X$ is irreducible, and $\tilde{S}$ is not degenerate on the positive orthant. In fact, the inequality holds for all $(t,s,i) \in D$, as $\phi$ is continuous, for every $(t,s,i)\in D$.
\end{proof}
\begin{remark}[Far boundary estimate]\noindent From Theorem \ref{theorem:lineargrowth}, we can obtain the error bound $\norm{\phi-\psi}_{V_1}$, where $\phi$ and $\psi$ are solutions of (\ref{eq:main_ivp})-(\ref{eq:operator}) and (\ref{eq:main_ibvp})-(\ref{eq:main_ibvp_boun}) respectively, and 
$$V_1= \left(C((0,T)\times \Gamma \times \mathcal{X} ), \|\phi\|_{V_1} = \sup_{t,s,i} \frac{|\phi(t,s,i)|}{1+\|s\|_1}\right).$$
To be more precise, the maximum error on the boundary due to the imposition of artificial data is not more than
$ \max\Big\{\| k_1e^{-(\min_{i} r(i))(T-t)}+k_2\cdot  s -h(t,s,i) \|_{V_1}, \| -k_3 e^{-(\max_{i} r(i))(T-t)}+ k_4\cdot  s  -h(t,s,i) \|_{V_1}\Big\}.$
In literature, this bound is often termed as \textit{far field boundary error estimate}.
\end{remark}
\subsection{Near field estimates} \label{ssec:near_field_estimate} 
In this subsection, we establish a few intermediate results for developing our final result Theorem \ref{theorem:pointwise_estimate}, which is an extension of \cite[Theorem 4]{KaN00}. The following lemma, which resembles to \cite[Lemma 1]{KaN00}, is proved here using a probabilistic method instead of an analytical approach. While \cite[Lemma 1]{KaN00} is for Black-Scholes-Merton model and deals with a scalar equation, the following lemma is for a parabolic system of equations originating from the regime-switching extension of the Black-Scholes-Merton model.
\begin{lemma}\label{lemma:far_feild_lemma1}
Let $f_1$ and $f_2$ be in $$C\left([0, T] \times (\Bar{\mathcal{U}} \cap (0,\infty)^d) \times \mathcal{X}\right) \cap C^{1,2}\left((0, T) \times \mathcal{U} \times \mathcal{X}\right),$$ with at most linear growth in space variable, where $\mathcal{U}$ is an open domain in $(0,\infty)^d$. We also assume that 
\begin{align}\label{eq:far_field_3}
\left(\frac{\partial f_1}{\partial t}+\mathbb{L}f_1 \right)(t,s,i)\leq 0 \quad \mbox{on} \quad (0,T) \times \mathcal{U} \times \mathcal{X},\\
\label{eq:far_field_4}
\left(\frac{\partial f_2}{\partial t}+\mathbb{L}f_2 \right)(t,s,i)\geq 0 \quad \mbox{on} \quad (0,T)\times \mathcal{U} \times \mathcal{X},
\end{align}
and $f_1 \geq f_2$ at $t=T$, and on $(0, T) \times (\partial\mathcal{U}\cap (0,\infty)^d) \times \mathcal{X}$. Then $f_1 \geq f_2$ in $(0,T) \times \mathcal{U} \times \mathcal{X}$.
\end{lemma}
\begin{proof} Given any two real numbers $a,b$, let $a \wedge b$ denote $\min(a, b)$. We fix a $(t,s,i)\in (0,T) \times \mathcal{U} \times \mathcal{X}$. We define $\tau=T\wedge \tau'$, where $\tau':=\inf\{t' \ge t | \tilde{S}(t') \notin{\mathcal{U}}\}$, gives the exit time of $\{\tilde{S}(t')\}_{t'\ge t}$ from $\mathcal{U}$. We further specify that $\tilde{S}$ solves the SDE in \eqref{eq:sde2} with $\tilde{S}(t) =s$. The transition kernel of the Markov chain $\{X_{t'}\}_{t'\ge t}$ is as before with $X_t=i$. Let $f_j$ be in $C\left([0, T] \times (\Bar{\mathcal{U}} \cap (0,\infty)^d) \times \mathcal{X}\right) \cap C^{1,2}\left((0, T) \times \mathcal{U} \times \mathcal{X}\right)$, with 
$$\sup \abs{\frac{f_j(t',s',i')}{1+\norm{s'}_1}} < \infty, \quad (t',s',i')\in [0,T] \times (\Bar{\mathcal{U}} \cap (0,\infty)^d) \times \mathcal{X},$$ 
for $j=1,2$. For each $j$ let the processes $N^{f_j}:=\{N^{f_j}_{t'}\}_{t'\ge t}$ be defined as in \eqref{eq:poisson} by 
$$ N^{f_j}_{t'}=e^{-\int_{t}^{t'\wedge \tau'}r(X_u)du}f_j(t'\wedge \tau', \tilde{S}(t'\wedge \tau'),X_{t'\wedge \tau'})
.$$
Then as in the proof of Lemma \ref{lemma:Ntmartingale}, we obtain that
\begin{align*}
t'\mapsto M^{f_j}_{t'} := N^{f_j}_{t'} - \int_t^{t'\wedge \tau'} e^{-\int_{t}^{u} r(X_{u'})du'} \left(\frac{\partial f_j}{\partial t}+\mathbb{L}f_j\right)(u,\tilde{S}(u),X_{u})du,
\end{align*}
is a local martingale for each $j=1,2$. 
Let us introduce a sequence of stopping times $\{\tau_n\}_n$, where $\tau_n$ represents the exit time of $\tilde{S}$ (starting from $s$ at time $t$) from an open neighbourhood of $s$ where the modulus of the functions $f_1$, $f_2$, and their first order time derivative, all first and second order partial space derivatives are bounded by $n$. Thus in the expression of $M^{f_j}_{t' \wedge \tau_n}$, the boundedness of $f_j$ and its partial derivatives may be assumed. Hence, by an argument similar to that appearing in the proof of Lemma \ref{lemma:Ntmartingale}, $\{M^{f_j}_{t' \wedge \tau_n}\}_{t'\ge t}$ is a martingale for each $n$. We also note that for every $t'\ge t$, $E[M^{f_j}_{t' \wedge \tau_n}]=E[M^{f_j}_{t\wedge \tau_n}]$. But $M^{f_j}_{t\wedge \tau_n}=M^{f_j}_{t}= N^{f_j}_{t}=f_j(t,s,i)$. So, for each $j$
\begin{equation}
\nonumber E(N^{f_j}_{t'\wedge \tau_n}) = f_j (t,s,i)+ E\left[\int_{t}^{t'\wedge \tau'\wedge \tau_n} e^{-\int_{t}^{u} r(X_{u'})du'} \left(\frac{\partial f_j}{\partial t}+\mathbb{L}f_j\right)(u,\tilde{S}(u),X_{u})du\right].
\end{equation}
The difference of above equations for $j=1$ and $2$ gives
\begin{align} \label{eq:far_field_ito3}
E\left[e^{-\int_{t}^{t'\wedge \tau' \wedge \tau_n}r(X_u)du}(f_1-f_2) (t'\wedge \tau'\wedge \tau_n, \tilde{S}(t'\wedge \tau'\wedge \tau_n),X_{t'\wedge \tau'\wedge \tau_n})\right] \hspace*{0.8in}\nonumber \\
=(f_1-f_2) (t,s,i)+E\Bigg[\int\limits_{t}^{t'\wedge \tau'\wedge \tau_n} e^{-\int_{t}^{u} r(X_{u'})du'} \left \{\left(\frac{\partial f_1}{\partial t}+\mathbb{L}f_1\right)-\left(\frac{\partial f_2}{\partial t}+\mathbb{L}f_2\right)\right\} \nonumber\\
(u,\tilde{S}(u),X_{u})du \Bigg] \leq (f_1-f_2)(t,s,i),
\end{align}
as $f_1$ and $f_2$ satisfy (\ref{eq:far_field_3}) and (\ref{eq:far_field_4}) respectively. On the other hand, $\tau_n\to \infty$ almost surely as $n\to \infty$. Therefore, due to the growth constraint on $f_j$ and \eqref{DMI}, the left side of (\ref{eq:far_field_ito3}) converges as $t' \uparrow T$ and $n\to \infty$ to 
$
E\left[e^{-\int_{t}^{T\wedge \tau'} r(X_u)du} (f_1-f_2) (T\wedge \tau',\right. $ $ \left.\tilde{S}(T\wedge \tau'),X_{T\wedge \tau'})\right],
$
which is non-negative due to the assumption that $f_1 \geq f_2$ at $t=T$, and on $(0, T) \times \partial\mathcal{U} \times \mathcal{X}$. Thus the right side of (\ref{eq:far_field_ito3}), i.e., $(f_1-f_2)(t,s,i)$ is non-negative too for every fixed $(t,s,i)\in (0,T) \times \mathcal{U} \times \mathcal{X}$.
\end{proof}
\par By following \cite{KaN00}, we introduce a parameterized function that satisfies a relevant partial differential inequality for certain choices of the parameters. In Lemmas \ref{lemma:near_feild1} and \ref{lemma:near_feild2}, we specify those ranges of parameters along with detailed proofs. For each $l\in \{1,2,...,d\}$, consider scalars $\epsilon_l > 0$, $\gamma_l\geq 0$, and $k_l \geq 1$ and a function $y_l:[0,T]\times R\times \mathcal{X} \rightarrow \mathbb{R}^+$ such that
\begin{align}\label{eq:yl}
y_{l}(t,s,i)=\frac{1}{\sqrt{T+\epsilon_l-t}}\exp{\left[-\gamma_l \left(\ln\frac{s_l}{k_ls_l^u}\right)^2/(T+\epsilon_l-t)\right]}, 
\end{align}
 which essentially depends only on time and the $l$th component of space variable.
 Since, this is constant on $\mathcal{X}$, the term $\sum_{j=1}^k\lambda_{ij}y_l(t,s,j)$, which appears in $\mathbb{L}y_l(t,s,i)$, is zero because the row sums are zero in $\Lambda=[\lambda_{ij}]$. Furthermore, we have
\begin{align}\label{eq:yl_deri}
\left\{
\begin{array}{rl}
\frac{\partial y_{l}}{\partial t}(t,s,i)&=\frac{1}{(T+\epsilon_l-t)^2}\left[\frac{T+\epsilon_l-t}{2}-\gamma_l\left(\ln\frac{s_l}{k_ls_l^u}\right)^2\right]y_{l}(t,s,i),\\
\frac{\partial y_{l}}{\partial s_l}(t,s,i)&=\frac{-2\gamma_l}{s_l(T+\epsilon_l-t)}\left[\ln\frac{s_l}{k_ls_l^u}\right]y_{l}(t,s,i),\\
\frac{\partial^2 y_{l}}{\partial s^2_l}(t,s,i)&=\frac{2\gamma_l}{s^2_{l}(T+\epsilon_l-t)^2}\left[\left(\ln\frac{s_l}{k_ls_l^u}-1\right)(T+\epsilon_l-t)+2\gamma_l\left(\ln\frac{s_l}{k_ls_l^u}\right)^2\right]y_l(t,s,i).
\end{array}\right.
\end{align}
We recall from (\ref{eq:operator}) that $a_{ll}(i)=\sum_{j=1}^d\sigma_{lj}^2(i)$. For each $l\in \{1,2,...,d\}$, we set 
\begin{equation}\label{eq:dl}
D_l:=\min\{a_{ll}(i)-2r(i): i \in \mathcal{X}\}.
\end{equation}
\begin{lemma}\label{lemma:near_feild1}
Fix an $l\in \{1,2,...,d\}$. If $D_l$ as in (\ref{eq:dl}) is positive, we set
\begin{align}\label{eq:parameters}
\gamma_l:=\frac{1}{2\ \underset{i}{\max}\{a_{ll}(i)\}}, \textrm{ and } k_l:=\exp{\left(\frac{\underset{i}{\max}\{a_{ll}(i)\}}{D_l}\right)}.
\end{align}
Then we have on $(0,T)\times R\times \mathcal{X}$ for any $\epsilon_l>0$
\begin{equation}\label{eq:near_field_inequality}
\frac{\partial y_l}{\partial t}(t,s,i)+\mathbb{L}y_l(t,s,i) \leq 0.\end{equation}
\end{lemma}
\begin{proof}
Using the expressions in  \eqref{eq:yl_deri}, the LHS of (\ref{eq:near_field_inequality}) becomes $y_l$ times the following term
\begin{align}\label{eq:lemma_18_lessthan}
&\gamma_l\left(2a_{ll}(i)\gamma_l-1\right)\left(\ln \frac{s_l}{k_l s^u_l}\right)^2+\gamma_l\left(a_{ll}(i)-2r(i)\right)(T+\epsilon_l-t)\left(\ln \frac{s_l}{k_l s^u_l}\right)\\
&+\left(\frac{1}{2}-\gamma_la_{ll}(i)-r(i)(T+\epsilon_l-t)\right)(T+\epsilon_l-t),\nonumber
\end{align}
for all $(t,s,i)\in (0,T)\times R\times \mathcal{X}$. After substituting $\gamma_l$, the above expression becomes
\begin{align*}
&\frac{1}{2\underset{i}{\max}\{a_{ll}(i)\}}\left(\frac{2a_{ll}(i)}{2\underset{i}{\max}\{a_{ll}(i)\}}-1\right)\left(\ln \frac{s_l}{k_l s^u_l}\right)^2+\frac{\left(a_{ll}(i)-2r(i)\right)}{2\underset{i}{\max}\{a_{ll}(i)\}}(T+\epsilon_l-t)\\&\left(\ln \frac{s_l}{k_l s^u_l}\right)+\left(\frac{1}{2}-\frac{a_{ll}(i)}{2\underset{i}{\max}\{a_{ll}(i)\}}-r(i)(T+\epsilon_l-t)\right)(T+\epsilon_l-t).
\end{align*}
By substituting $k_l$ in the second additive term of the above expression, we get
\begin{align*}
&\frac{1}{2\underset{i}{\max}\{a_{ll}(i)\}}\left(\frac{2a_{ll}(i)}{2\underset{i}{\max}\{a_{ll}(i)\}}-1\right)\left(\ln \frac{s_l}{k_l s^u_l}\right)^2+\frac{\left(a_{ll}(i)-2r(i)\right)}{2\underset{i}{\max}\{a_{ll}(i)\}}(T+\epsilon_l-t)\\
&\times \left(\ln s_l-\ln s^u_l\right) -\frac{\left(a_{ll}(i)-2r(i)\right)}{2\underset{i}{\max}\{a_{ll}(i)\}}(T+\epsilon_l-t)\frac{\underset{i}{\max}\{a_{ll}(i)\}}{D_l}+\frac{1}{2}(T+\epsilon_l-t)\\
&-\frac{a_{ll}(i)}{2\underset{i}{\max}\{a_{ll}(i)\}}(T+\epsilon_l-t)-r(i)(T+\epsilon_l-t)^2.
\end{align*}
Note that except the fourth term, all other terms are non-positive because $(T+\epsilon_l-t)$ is always positive. Since $D_l>0$, it is clear that the third term dominates the fourth term in magnitude. Hence \eqref{eq:lemma_18_lessthan} is non-positive, and the result follows. 
\end{proof}

\begin{lemma}\label{lemma:near_feild2}
Let us fix $l\in \{1,\ldots, d\}$, and assume that $D_l\leq 0$. Fix a point $(\hat{t},\hat{s}) \in [0,T)\times R $ such that
\begin{align}\label{sl'}
\ln \frac{s^u_l}{\hat{s}_l}>-D_l(T-\hat{t}).
\end{align}
Using this point, we set the following values of the parameters
\begin{align} \label{parameters}
\epsilon_l=\frac{(T-\hat{t})\ln k_l}{\ln \frac{s^u_l}{\hat{s}_l}}, \quad \gamma_l < \min\left( \frac{1}{2\underset{i}{\max}\{a_{ll}(i)\}}\left(1+\frac{D_l\epsilon_l}{\ln k_l}\right), \frac{\epsilon_l}{2\ln k_l}\right).
\end{align}
Then let $y_l$ be as in \eqref{eq:yl}, then \eqref{eq:near_field_inequality} holds for all $(t,s,i)\in(0,T)\times R\times \mathcal{X}$, with these parameter values and for sufficiently large $k_l$.
\end{lemma}
\begin{proof} First, we argue that $(\hat{t},\hat{s})$ satisfying (\ref{sl'}) exists in $[0,T)\times R $. The left side of \eqref{sl'} is clearly positive for any $\hat{s}\in R$. Hence, one may choose $\hat{t}$ sufficiently closer to $T$ so that \eqref{sl'} holds. As shown in the proof of Lemma \ref{lemma:near_feild1}, we need to show that \eqref{eq:lemma_18_lessthan} is non-positive. We first consider the first term of (\ref{eq:lemma_18_lessthan}). Since  $D_l\leq 0$ and $T-\hat{t}>0$, using \eqref{sl'}, $0< \ln \frac{s^u_l}{\hat{s}_l}+D_l(T-\hat{t}) \le \ln \frac{s^u_l}{\hat{s}_l}$. Thus we have 
$ 0<1+\frac{D_l(T-\hat{t})}{\ln \frac{s^u_l}{\hat{s}_l}} \le 1$ as $\ln \frac{s^u_l}{\hat{s}_l}$ is positive. 
Therefore, a non-negative $\gamma_l$ can be chosen satisfying \eqref{parameters}, and hence 
$ 0\le a_{ll}(i)\gamma_l < 1/2$. That is, $(2a_{ll}(i)\gamma_l-1) < 0$. Hence the first term of (\ref{eq:lemma_18_lessthan}) is negative. We next consider the coefficient of $r(i)$ from second and third term of (\ref{eq:lemma_18_lessthan}) and simplify as $-2\gamma_l(T+\epsilon_l-t)\ln \frac{s_l}{k_ls^u_l}-(T+\epsilon_l-t)^2$, that is $\left[2\gamma_l\ln \left(\frac{k_ls^u_l}{s_l}\right)-(T+\epsilon_l-t)\right](T+\epsilon_l-t)$. By substituting the value of $\epsilon_l$ from \eqref{parameters}, the above is negative provided 
$$2\gamma_l\ln\frac{s^u_l}{s_l} -(T-t) + \left( 2\gamma_l - (T-\hat{t})/\ln \frac{s^u_l}{\hat{s}_l}\right) \ln k_l < 0.$$
We note that due to \eqref{parameters}, the multiplier of $\ln k_l$ in the above expression is negative for all $(t,s,i)\in(0,T)\times R\times \mathcal{X}$. Therefore, for a sufficiently large  $k_l$ the above inequality holds, i.e., the coefficient of $r(i)$ is negative. Next, we rewrite the coefficient of $a_{ll}(i)$ from second and third term of (\ref{eq:lemma_18_lessthan}) as  
\begin{align*}
   \gamma_l(T+\epsilon_l-t)\left( \ln\left(\frac{s_l}{k_ls^u_l}\right)-1\right),
\end{align*}
which is negative for all $k_l\ge 1$ irrespective to the magnitude of $\gamma_l$ for all $(t,s,i)\in(0,T)\times R\times \mathcal{X}$. Furthermore, by using the expression of $\epsilon_l$ we see that the magnitude of the above term grows as $\mathcal{O}(\ln k_l)^2$ as $k_l\to \infty$. Hence this dominates the only other remaining term in \eqref{eq:lemma_18_lessthan}, i.e., $\frac{1}{2}(T+\epsilon_l-t)$
which grows as $\mathcal{O}(\ln k_l)$. Thus for sufficiently large $k_l$ the value of the expression in \eqref{eq:lemma_18_lessthan} is negative  for all $(t,s,i)\in(0,T)\times R\times \mathcal{X}$.
\end{proof}
\begin{figure}[h!]
    \centering
    \includegraphics[width=0.3\textwidth]{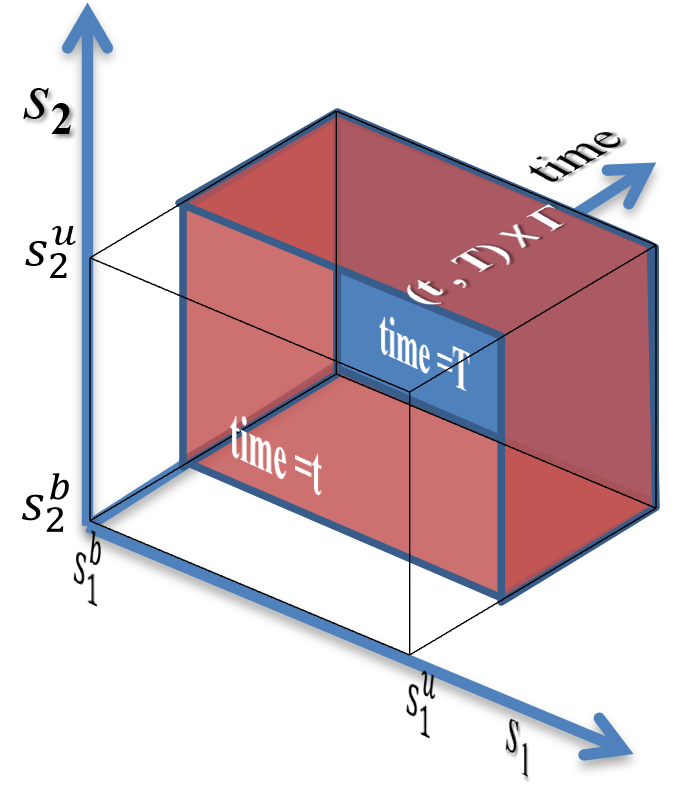}
    \caption{Illustration of $(t, T ) \times \Gamma$ for $d=2$.}
    \label{fig.boundary}
\end{figure}
\noindent Next, we obtain a point-wise error estimate at an interior point of the domain by considering a general $y_l$ satisfying \eqref{eq:near_field_inequality}.
Due to this generality, the error bound obtained in the following theorem is abstract in nature, and is not expressed explicitly in terms of the model parameters. This limitation is addressed in Theorem \ref{theorem:pointwise_estimate}. The point-wise error is to be estimated in terms of the maximum boundary error where the supremum should be taken over $(t, T ) \times \Gamma \times \mathcal{X}$. This set has been illustrated for a single-regime and two-dimension case in Figure \ref{fig.boundary}.

\begin{theorem}\label{theorem:near_feild_1}
Let $v$ and $w$ satisfy (\ref{eq:main_ibvp}) and (\ref{eq:main_ibvp_ini}). Let $y:[0,T]\times R\times \mathcal{X}\rightarrow \mathbb{R}^+$ be given by  $y(t,s,i):=\sum_{l=1}^dC_ly_l(t,s,i)$, $C_l\geq0$, where for each $l$, $y_l$ satisfies  \eqref{eq:near_field_inequality} on $(0,T)\times R\times \mathcal{X}$. Then for $(t,s,i)\in (0,T)\times R \times \mathcal{X}$ we have 
\begin{equation}
\abs{v(t,s,i)-w(t,s,i)}\leq \sup_{(t',s',i') \in (t,T)\times \Gamma \times \mathcal{X}}\left[\frac{\abs{v(t',s',i')-w(t',s',i')}}{y(t',s',i')}\right]y(t,s,i).
\end{equation}
\end{theorem}
\begin{proof}
For a fixed $0<t<T$, let us define
\begin{equation}\label{eq:f_def}
f(\tau,s,i):=\sup_{(t',s',i') \in (t,T)\times \Gamma \times \mathcal{X}}\left[\frac{\abs{v(t',s',i')-w(t',s',i')}}{y(t',s',i')}\right]y(\tau,s,i),
\end{equation}
for all $(\tau,s,i)\in (t,T)\times R \times \mathcal{X}$. On this domain we have using (\ref{eq:near_field_inequality}),
\[
\frac{\partial f}{\partial \tau}+\mathbb{L}f=\sup_{(t',s',i') \in (t,T)\times \Gamma \times \mathcal{X}}\left[\frac{\abs{v(t',s',i')-w(t',s',i')}}{y(t',s',i')}\right]\sum_{l=1}^d C_l \left(\frac{\partial y_l}{\partial \tau}+\mathbb{L}y_l\right) \leq 0,
\]
as $C_l\ge 0$ and $y$ is positive. On the other hand $v-w$, satisfies 
\begin{align*}
\frac{\partial (v-w)}{\partial \tau}+\mathbb{L}(v-w)= 0,\:\:\forall \:\:(\tau,s,i)\in (t,T)\times R \times \mathcal{X}.
\end{align*}
Hence, by Lemma \ref{lemma:far_feild_lemma1}, we get 
\begin{align}\label{fy}
f(\tau,s,i) \geq (v-w)(\tau,s,i) \ \ \forall \ \ (\tau,s,i)\in (t,T)\times R \times \mathcal{X},
\end{align}
as $f \geq 0= (v-w)$ on $\tau=T$ as well as due to (\ref{eq:f_def})  $f(\tau,s,i) \ge \abs{v(\tau,s,i)-w(\tau,s,i)}\ge v(\tau,s,i)-w(\tau,s,i)$ for any $(\tau,s,i) \in (t,T)\times \Gamma \times \mathcal{X}$. Now we observe that the above argument follows if we replace $v-w$ by $w-v$. Thus as in \eqref{fy} we obtain $f\ge w-v$ on $(t,T)\times R \times \mathcal{X}$. Therefore by combining this and \eqref{fy} we get $f\ge |v-w|$ on $(t,T)\times R \times \mathcal{X}$. 
As $t\in (0,T)$ has been fixed arbitrarily,   $f(t,s,i) \geq |(v-w)|(t,s,i)$ for any $0< t <T$, $s\in R$ and $i\in \mathcal{X}$.
\end{proof}
\par Lemmas \ref{lemma:near_feild1} and \ref{lemma:near_feild2} show that for any values of model parameters, one can suitably define a function $y_l$ on $[0,T]\times R\times \mathcal{X}$ as in \eqref{eq:yl} satisfying \eqref{eq:near_field_inequality}. With the help of such functions $y_1, \ldots, y_d$, a point-wise error estimate at an interior point is derived in terms of the maximum error on the far boundary in the following theorem. Unlike in Theorem \ref{theorem:near_feild_1}, this estimate is explicitly expressed in terms of the model parameters. In this connection, we recall that Theorem 4 in \cite{KaN00}, also gives an error estimate serving a similar purpose. It also turns out that the consideration of regime switching extension does not prevent one to derive an expression identical to that in \cite{KaN00} by mimicking the arguments therein. However, that estimate works only on a strictly smaller subdomain satisfying \eqref{sl'}. So, we propose a different estimate that works globally. Hence, the expressions and derivations of these two estimates are significantly different.
\begin{theorem}\label{theorem:pointwise_estimate}
Let $v$ and $w$ be the classical solutions of (\ref{eq:main_ivp})-(\ref{eq:main_ivp_1}) and (\ref{eq:main_ibvp})-(\ref{eq:main_ibvp_boun}) respectively. Then at each point $(t',s',i')\in [0,T]\times R \times \mathcal{X}$, we have
\begin{align}\label{eq:vminusw}
\nonumber &\abs{v(t',s',i')-w(t',s',i')} \le  \sup_{[t',T)\times \Gamma\times \mathcal{X}}   \abs{v-w} \\
& \times \sum_{l=1}^d \exp \left(\frac{ - \ln(\frac{s_l^u}{s'_l}) \left(\frac{D_l^+}{\underset{i}{\max} \{a_{ll}(i)\}} \ln(\frac{s_l^u}{s'_l})+2\right)
+ (\underset{i}{\max} \{a_{ll}(i)\}+|D_l|)(T-t')}{2\left(D_l^+ (T-t')+\frac{\underset{i}{\max} \{a_{ll}(i)\}}{(\underset{i}{\max} \{a_{ll}(i)\}+D_l^+)}\right)}\right),
\end{align}
where $D_l^+ =\max\{D_l, 0\}$.
\end{theorem}
\begin{proof} 
For each $l$, let $y_l$ be as in \eqref{eq:yl} satisfying \eqref{eq:near_field_inequality} on $(0,T)\times R \times \mathcal{X}$. We also denote the far facet of $R$ in the $l$th direction by $\Gamma_l:= \{s\in \Gamma : s_l=s^u_{l}\}$. We note that if for each $t'\in (0,T)$, $C_l(t')$ is defined as $\frac{\sup_{[t',T)\times \Gamma_l\times \mathcal{X}} \abs{v-w}}{\inf_{ (t',T) \times \Gamma_l\times \mathcal{X}}y_l}$,
then for all $t\in [t',T)$ and $s\in \Gamma$, $\abs{v(t,s,i)-w(t,s,i)}$ is less than or equal to
$\sup_{[t',T)\times \Gamma\times \mathcal{X}} \abs{v-w} \le \sum_{l=1}^d \sup_{[t',T)\times \Gamma_l\times \mathcal{X}} \abs{v-w}\le  \sum_{l=1}^d C_l(t') y_l(t,s,i).$
Then from Lemma {\ref{lemma:far_feild_lemma1}}, 
we write for all $(t,s,i)\in (t',T)\times R \times \mathcal{X}$
\begin{align*}
\abs{v(t,s,i)-w(t,s,i)}\leq  \sum_{l=1}^d C_l(t')y_l(t,s,i)
\leq & \sum_{l=1}^d \sup_{[t',T)\times \Gamma_l \times \mathcal{X}} \abs{v-w} \frac{y_l(t,s,i)}{\inf_{ (t',T) \times \Gamma_l\times \mathcal{X}}y_l},\\
\leq & \sup_{[t',T)\times \Gamma\times \mathcal{X}}   \abs{v-w} \sum_{l=1}^d \frac{y_l(t,s,i)}{\inf_{ (t',T) \times \Gamma_l\times \mathcal{X}}y_l}.
\end{align*}
Next, using the continuity in $t$ variable and allowing $t\downarrow t'$ on both sides, we write
\begin{align}\label{eq:vminusw_1}
\abs{v(t',s',i')-w(t',s',i')}\leq \sup_{[t',T)\times \Gamma\times \mathcal{X}}   \abs{v-w} \sum_{l=1}^d Y_l(t',s',i'),
\end{align}
where  $Y_l(t',s',i') $ is defined as
\begin{align*}
&\inf\left\{ \frac{y_l(t',s',i')}{\inf_{ (t',T) \times \Gamma_l\times \mathcal{X}}y_l}  \mid \textrm{ for each } l, \ y_l \textrm{ is given by} \eqref{eq:yl} \textrm{and satisfies } \eqref{eq:near_field_inequality} \right\}.
\end{align*}
Next a simple expression for an upper bound of $Y_l$ (defined in \eqref{eq:vminusw_1}) is derived. It is evident from \eqref{eq:yl_deri} that if there is a $t_0\in [0,T)$ such that the sign of $\frac{\partial y_l}{\partial t}(t_0, s)$ is non-positive, $\frac{\partial y_l}{\partial t}\le 0$ for all $t\ge t_0$. Hence, $y_l(\cdot,s)$ attains its minimum on $[t',T]$ either at $t=t'$ or $t=T$. Therefore using \eqref{eq:yl} and by simplifying some terms we get
\begin{align}\label{eq:hl}
\nonumber  H_l(t',s'_l,i'; \epsilon_l, \gamma_l,& k_l) :=\frac{y_l(t',s',i')}{\inf_{ (t',T) \times \Gamma_l\times \mathcal{X}}y_l},\\
\nonumber & =\max \left\{\exp\left(-\frac{\gamma_l}{T+\epsilon_l-t'} \ln \frac{s_l^u}{s'_l}
\ln\left(\frac{k_l^2 s_l^u}{s'_l}\right)\right),\right.\\
\nonumber & \hspace{0.4in}\left.\sqrt{\frac{\epsilon_l}{T+\epsilon_l-t'}}\exp\left(-\frac{\gamma_l}{T+\epsilon_l-t'}\ln^2\left(\frac{k_l s_l^u}{s'_l}\right)+\frac{\gamma_l}{\epsilon_l}\ln^2 k_l \right)\right\},\\
& \leq \exp \left(-\gamma_l\frac{\ln^2\left(k_l\frac{s_l^u}{s'_l}\right)}{(T+\epsilon_l-t')}+\frac{\gamma_l}{\epsilon_l}\ln^2 k_l\right),
\end{align}
where $\gamma_l$, $k_l$ and $\epsilon_l$ are as in Lemmas \ref{lemma:near_feild1} or  \ref{lemma:near_feild2} depending on the model parameters. The above inequality is obtained by observing that for any $\epsilon_l>0$, $\frac{\epsilon_l}{T+\epsilon_l-t'}<1$ and 
$$
\frac{-\gamma_l}{T+\epsilon_l-t'}\left[\ln\frac{s_l^u}{s'_l}\ln \left(\frac{k_l^2 s_l^u}{s_l'}\right)-\ln^2 \left(\frac{k_l s_l^u}{s_l'}\right)\right]=\frac{\gamma_l \ln^2 k_l}{T+\epsilon_l-t'}< \frac{\gamma_l \ln^2 k_l}{\epsilon_l}.
$$
Next, we substitute the parameters $\epsilon_l$, $\gamma_l$, and $k_l$ in \eqref{eq:hl} by not violating the constraints in Lemmas \ref{lemma:near_feild1}, and \ref{lemma:near_feild2} to obtain the upper bound \eqref{eq:vminusw} of $Y_l(t',s',i')$.
We first set $\epsilon_l=\frac{\ln k_l}{\underset{i}{\max} \{a_{ll}(i)\}+|D_l|}$ and get 
\begin{align*}
Y_l(t',s',i') \leq \exp\left(
-\gamma_l\frac{\ln \frac{s_l^u}{s'_l} \left( \ln \frac{s_l^u}{s'_l} + 2\ln k_l \right)-(\underset{i}{\max} \{a_{ll}(i)\}+|D_l|)\ln k_l(T-t')}{(T-t')+\frac{\ln k_l}{\underset{i}{\max} \{a_{ll}(i)\}+|D_l|)}}
\right),
\end{align*}
because
\begin{align*}
&\frac{-\gamma_l \ln^2\left(\frac{ k_ls_l^u}{s'_l}\right)}{T-t'+\frac{\ln k_l}{\underset{i}{\max} \{a_{ll}(i)\}+|D_l|}}+\frac{\gamma_l {(\underset{i}{\max} \{a_{ll}(i)\}+|D_l|})\ln^2 k_l}{\ln k_l}\\
&=-\gamma_l\frac{\ln^2\left(\frac{k_ls_l^u}{s'_l}\right)-(\underset{i}{\max} \{a_{ll}(i)\}+|D_l|)\ln k_l(T-t')-\ln^2 k_l}{(T-t')+\frac{\ln k_l}{\underset{i}{\max} \{a_{ll}(i)\}+|D_l|)}},\\
&=-\gamma_l\frac{\ln \frac{s_l^u}{s'_l} \ln\left(\frac{k_l^2 s_l^u}{s'_l}\right)-(\underset{i}{\max} \{a_{ll}(i)\}+|D_l|)\ln k_l(T-t')}{(T-t')+\frac{\ln k_l}{\underset{i}{\max} \{a_{ll}(i)\}+|D_l|)}}.
\end{align*}
\noindent {\bf Case $1 (D_l>0):$} Using \eqref{eq:parameters}, i.e., $\ln k_l=\frac{\underset{i}{\max} \{a_{ll}(i)\}}{D_l}$ and $\gamma_l=\frac{1}{2 \underset{i}{\max} \{a_{ll}(i)\}}$, we get
\begin{align}
Y_l(t',s',i') \leq  
\exp \left(-\frac{\ln\frac{s_l^u}{s'_l} \left( \ln \frac{s_l^u}{s'_l} + 2\frac{\underset{i}{\max} \{a_{ll}(i)\}}{D_l}\right)
-(\underset{i}{\max} \{a_{ll}(i)\}+D_l)\frac{\underset{i}{\max} \{a_{ll}(i)\}}{D_l}(T-t')}{2\underset{i}{\max} \{a_{ll}(i)\}\left((T-t')+\frac{\underset{i}{\max} \{a_{ll}(i)\}}{D_l (\underset{i}{\max} \{a_{ll}(i)\}+D_l)}\right)}\right),\nonumber
\end{align}
\begin{align} \label{eq:case1}
= & \exp \left(\frac{ - \ln\frac{s_l^u}{s'_l} \left( \frac{D_l}{\underset{i}{\max} \{a_{ll}(i)\}} \ln\frac{s_l^u}{s'_l} + 2 \right)
+  (\underset{i}{\max} \{a_{ll}(i)\}+D_l)(T-t')}{2\left(D_l (T-t')+\frac{\underset{i}{\max} \{a_{ll}(i)\}}{(\underset{i}{\max} \{a_{ll}(i)\}+D_l)}\right)}\right).
\end{align}

\noindent {\bf Case $2 (D_l\leq 0):$} The above choice of $\epsilon_l$ and \eqref{parameters} imply that
$0<\frac{T-\hat{t}}{\ln \frac{s_l^u}{\hat{s}_l}}=\frac{1}{\underset{i}{\max} \{a_{ll}(i)\}-D_l}< \frac{1}{-D_l}.$
Hence \eqref{sl'} holds. Consequently, \eqref{parameters} implies that $\gamma_l$ may be taken from $(0,\hat{\gamma_l})$, where $\hat{\gamma_l}=\frac{1}{2\left(\underset{i}{\max} \{a_{ll}(i)\}-D_l\right)} >0$. Then
$$
Y_l(t',s',i') \leq \exp \left(-\hat{\gamma_l}\frac{\ln\frac{s_l^u}{s'_l} \left(\left(\ln\frac{s_l^u}{s'_l}\right)/\ln k_l + 2 \right)-\left(\underset{i}{\max} \{a_{ll}(i)\}-D_l\right)(T-t')}{(T-t')/\ln k_l + 1/(\underset{i}{\max} \{a_{ll}(i)\}-D_l)}\right).
$$
Since, Lemma \ref{lemma:near_feild2} holds for sufficiently large $k_l$, letting $k_l \to \infty$ in the above, we get
$
Y_l(t',s',i') \leq \exp \left(-\ln\left(\frac{s_l^u}{s'_l}\right)+ \left(\underset{i}{\max} \{a_{ll}(i)\}-D_l\right)(T-t')/2\right)$. By combining the inequalities appearing above and in \eqref{eq:case1}, we get for both the cases
\begin{align}\label{Ybound}
Y_l(t',s',i') \leq \exp \left(\frac{ - \ln(\frac{s_l^u}{s'_l}) \left(\frac{D_l^+}{\underset{i}{\max} \{a_{ll}(i)\}} \ln(\frac{s_l^u}{s'_l})+2\right)
+ (\underset{i}{\max} \{a_{ll}(i)\}+|D_l|)(T-t')}{2\left(D_l^+ (T-t')+\frac{\underset{i}{\max} \{a_{ll}(i)\}}{(\underset{i}{\max} \{a_{ll}(i)\}+D_l^+)}\right)}\right).
\end{align}
Hence, \eqref{eq:vminusw} follows from the above bound and \eqref{eq:vminusw_1}.
\end{proof}
\begin{remark}\label{remPsi}
In the preamble of Theorem \ref{theorem:pointwise_estimate}, we have mentioned the possibility of deriving an estimate $\Psi_l(t,s,i)$ of $Y_l$, that is valid on $\mathcal{D}:=\{(t,s,i)\in [0,T]\times R\times \mathcal{X} \mid \ln \frac{s^u_l}{s_l}+D_l(T-t)\ge 0, \ \forall l=1,\ldots, d\}$, by mimicking the approach of \cite{KaN00}. On the other hand a globally valid estimate $\bar \Psi_l(t,s,i)$ of $Y_l$ has been obtained in \eqref{Ybound}. Thus \eqref{eq:vminusw} may be improved as $\abs{v(t',s',i')-w(t',s',i')} \le  \sup_{[t',T)\times \Gamma\times \mathcal{X}}   \abs{v-w}\sum_{l=1}^d \hat \Psi_l(t',s',i')$ where,
$\hat \Psi_l:= \min\{ \Psi_l, \bar{\Psi}_l\} 1_{\mathcal{D}} + \bar{\Psi}_l 1_{([0,T]\times R\times \mathcal{X}) \setminus \mathcal{D}}$.
\end{remark}
\subsection{Numerical Study}
A comparison of two estimates, mentioned in Remark \ref{remPsi} is presented below by considering a couple of numerical examples. The estimate of $Y_l$ in \eqref{Ybound} for a single regime case may easily be compared with the estimate presented in \cite{KaN00}. We present the comparison for the single asset case, i.e., $d=1$. It will be shown that none dominates the other. We set $\mathcal{X}=\{1\}$, and denote $\sigma:= \sigma(1)$, $r:=r(1)$, $D:=D_1$. So, the estimates in \cite{KaN00} and \eqref{Ybound} are
\begin{align*}
\Psi_1(t,s_1,1)=\exp \left(-\frac{\ln(s^u_1/s_1)(\ln(s^u_1/s_1) + \min\{0,D\}(T-t))}{2 \sigma^2(T-t)}\right),\\
\bar{\Psi}_1(t,s_1,1)=\exp \left(\frac{ - \ln(s_1^u/s_1) \left(\frac{D^+}{\sigma^2} \ln(s_1^u/s_1)+2\right)
+ (\sigma^2+|D|)(T-t)}{2\left(D^+ (T-t)+\frac{\sigma^2}{(\sigma^2+D^+)}\right)}\right),
\end{align*}
respectively. Now we set $T=1$, $s^b_1 =0$, $s^u_1=20$. Figure \ref{fig1}  presents a contour plot of $\Psi_1-\bar{\Psi}_1$, where $\sigma =0.4$, and $r=1\%$. Since $D:= \sigma^2-2r=0.14>0$, both the estimates are valid on the full region. The contour plot in Figure \ref{fig1}, where $t$, and $s_1$ variables are along vertical and horizontal axes, shows that $\Psi_1-\bar{\Psi}_1$ takes both positive and negative values for this example. Hence  $\hat{\Psi}_1=\min\{\Psi_1 , \bar{\Psi}_1\}$ is strictly sharper than both $\Psi_1$ and $\bar{\Psi}_1$. 
\begin{figure}[h!]
     \centering
     \begin{subfigure}[b]{0.47\textwidth}
         \centering
         \includegraphics[width=\textwidth]{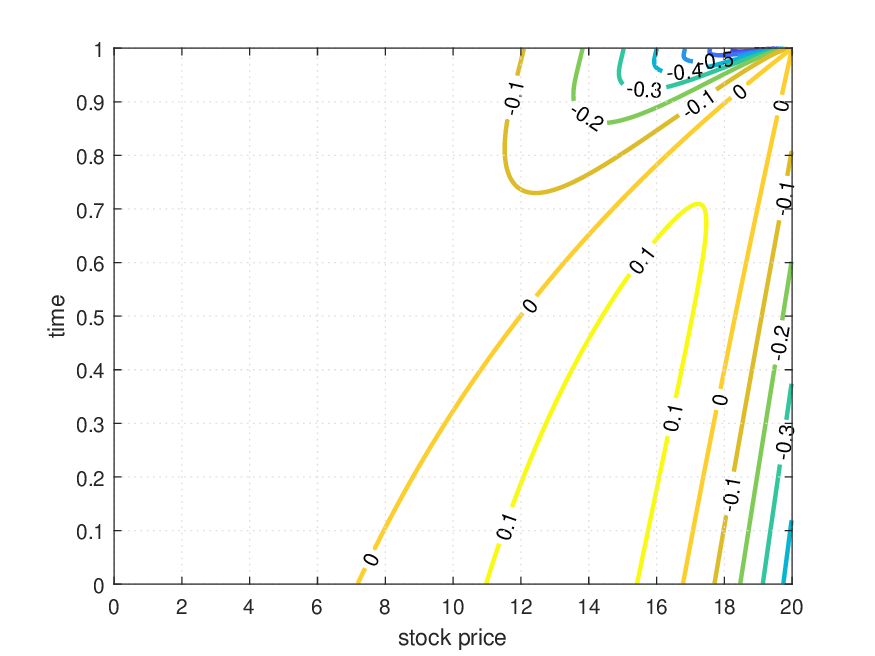}
         \caption{Contour plot of $\Psi_1-\bar{\Psi}_1$ where $\sigma =0.4$, and $r=0.01$.}
         \label{fig1}
     \end{subfigure}
     \begin{subfigure}[b]{0.5\textwidth}
         \centering
         \includegraphics[width=\textwidth]{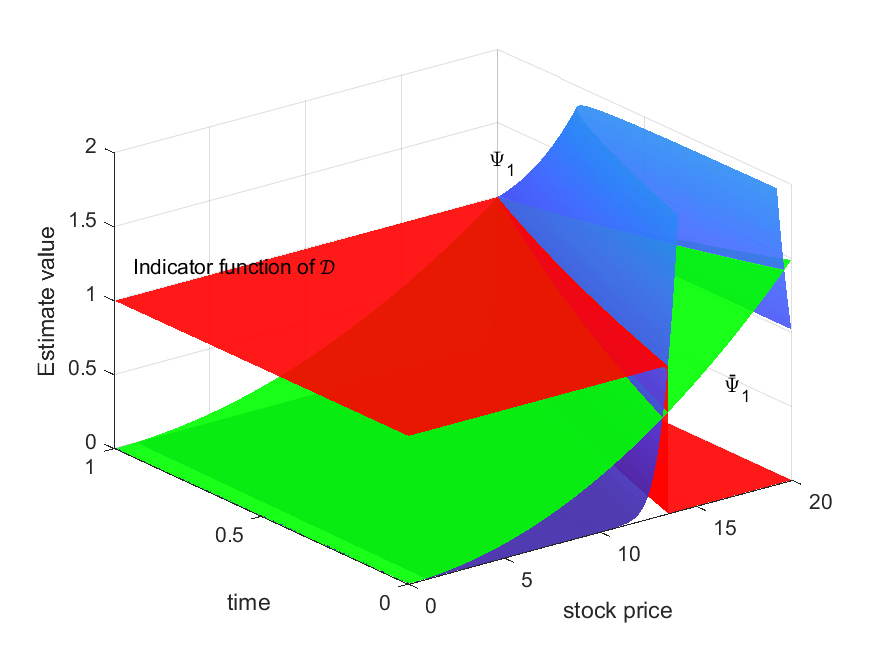}
         \caption{Surface plot of $\Psi_1$, $\bar{\Psi}_1$, and the indicator function of the domain $\mathcal{D}$ where $\sigma =0.1$, and $r=0.2$.}
         \label{fig2}
     \end{subfigure}
        \caption{Comparison of estimates.}
        \label{both_figures}
\end{figure}
Next by setting $\sigma = 0.1$ and $r=20\%$, we get $D=-0.39 <0$. Figure \ref{fig2} includes surface plots of $\Psi_1$, $\bar{\Psi}_1$, and indicator function of the domain $\mathcal{D}:=\{(t,s_1)\mid \ln \frac{s^u_1}{s_1}+D(T-t)\ge 0\}$, against $t$-$s_1$ plane. Those surfaces are colored in blue, green, and red respectively.
We observe that $\Psi_1-\bar{\Psi}_1$ changes its sign even inside $\mathcal{D}$, where $\Psi_1$ is a valid estimate for this example. Hence $\hat \Psi_1= \min\{ \Psi_1, \bar{\Psi}_1\} 1_{\mathcal{D}} + \bar{\Psi}_1 1_{\mathcal{D}^c}$ gives a strictly sharper estimate on the whole truncated domain.

\noindent The question regarding the  location for the artificial boundary depending on the error tolerance has been addressed in \cite{KaN00} by using the estimate discussed above. The analysis in \cite{KaN00}, that considers few realistic numerical examples, convinces of the usefulness of a result like Theorem 4.7. Therefore, in view of Remark \ref{remPsi}, we have significantly improved the applicability of the result presented in \cite{KaN00} by sharpening the estimate, extending the region of validity, and including the regime switching scenario.

\section{Conclusions}\label{sec:conclu}
Using the probabilistic method a self-contained  proof has been developed for the existence of a unique solution of a system of PDEs in the class of functions having at most linear growth. The system under consideration originates from the regime-switching extension of the Black-Scholes-Merton model. A growth estimate has been derived for the solution depending on the growth of the terminal data. This growth estimate has further been utilized in estimating the maximum error at the boundary due to the imposed boundary data of the truncated domain problem. An error estimate has also been obtained at every interior point of the domain. Finally, the estimate has been expressed using the model parameters and maximum error on the far boundary. These results are useful for allocating the artificial boundary depending on the error tolerance. For more details in this regard, \cite{KaN00} may be referred to. The error estimate obtained in that paper has been compared with the estimate of the present paper under the special case where regime-switching is absent. By considering realistic numerical examples, we show that our results significantly improve over the results presented in \cite{KaN00} by sharpening the estimate and extending the region of validity of the estimate. Moreover, we obtain point-wise estimates for domain truncation error of a system of PDE that arises in the option pricing problem under the regime-switching scenario. As a future direction, the proposed results may be extended to the option pricing equations with more sophisticated model assumptions. For example, a similar study for the regime-switching stochastic volatility models is absent in the literature.

\appendix
\section{Proofs of Lemmata in Section \ref{sec:exis_uniq}}\label{appendixA}
\begin{proof} [Proof of Lemma \ref{lemma:used_in_2}] Using the expectation of log-normal random variable, we can write 
\begin{align*}
\int_{(0,\infty)^d}\left(1+\sum_{l=1}^d x_l\right) \alpha(x,s,i,v)dx&=
1+\sum_{l=1}^d\int_{(0,\infty)^d} x_l \alpha(x,s,i,v)dx,\\ 
&=1+\sum_{l=1}^d e^{\left[\ln s_l +(r(i)-\frac{1}{2}a_{ll}(i))v+\frac{1}{2}a_{ll}(i)v\right]},\\
& = 1+\sum_{l=1}^d s_l e^{r(i)v}.
\end{align*}
\end{proof}
\begin{proof}[Proof of Lemma \ref{lemma:delalphadelv}]
From (\ref{eq:alpha}), we have
\begin{equation}\label{eq:alpha_things}
\left.\begin{aligned}
&\frac{\partial \tilde{z}_l}{\partial v}=r(i)-\frac{1}{2}a_{ll}(i), \quad \mathfrak{S}(i)=v a(i) \implies \frac{\partial \mathfrak{S}(i)}{\partial v}=a(i), \quad \mathfrak{S}(i)^{-1}=\frac{1}{v}a^{-1}(i),\\ & trace\left(\mathfrak{S}(i)^{-1}\frac{\partial \mathfrak{S}(i)}{\partial v}\right)=trace\left(\frac{1}{v}I\right)=\frac{d}{v},\\
 &\frac{\partial}{\partial v}(\mathfrak{S}(i)^{-1})_{ll'}=\frac{\partial}{\partial v}\left(\frac{1}{v}a^{-1}(i)\right)_{ll'}=\frac{-1}{v^2}\left(a^{-1}(i)\right)_{ll'}.
\end{aligned}\right\}
\end{equation}
Taking logarithm on both
the sides of (\ref{eq:alpha}), we get
\begin{align}\label{eq:logalpha}
\ln{\alpha(x,s,i,v)}=-\ln{\left(\frac{1}{\sqrt{(2\pi)^d}\:x_1.x_2...x_d}\right)}-\frac{1}{2}\ln{|\mathfrak{S}(i)|}-\frac{1}{2}\sum_{l=1}^d\sum_{l'=1}^d (\mathfrak{S}(i)^{-1})_{ll'}(z_l-\tilde{z}_l)(z_{l'}-\tilde{z}_{l'}).
\end{align}
Differentiating above expression (\ref{eq:logalpha}) w.r.t. $v$ and using (\ref{eq:alpha_things}) and Jacobi's formula for derivative of determinant, we can write
\begin{align}\label{eq:delalpha}
&\frac{\partial \alpha}{\partial v}= \left[ \frac{-d}{2v}-\frac{1}{2}\sum_{l=1}^d\sum_{l'=1}^d\left(\frac{-1}{v^2}(a^{-1}(i))_{ll'}(z_l-\tilde{z}_l)(z_{l'}-\tilde{z}_{l'})\right)+\frac{1}{2}\sum_{l=1}^d\sum_{l'=1}^d\frac{1}{v}(a^{-1}(i))_{ll'}\right.\nonumber\\
&\left.(z_{l'}-\tilde{z}_{l'})\left(r(i)-\frac{1}{2}a_{ll}(i)\right)+\frac{1}{2}\sum_{l=1}^d\sum_{l'=1}^d\frac{1}{v}(a^{-1}(i))_{ll'}\left(r(i)-\frac{1}{2}a_{l'l'}(i)\right)(z_l-\tilde{z}_l)\right]\alpha ,\nonumber\\
=&\left[ \frac{-d}{2v}+\frac{1}{2}\sum_{l=1}^d\sum_{l'=1}^d\left(\frac{1}{v^2}(a^{-1}(i))_{ll'}(z_l-\tilde{z}_l)(z_{l'}-\tilde{z}_{l'})\right)+\sum_{l=1}^d\sum_{l'=1}^d\frac{1}{v}(a^{-1}(i))_{ll'}(i)\right.\nonumber\\
&\left.(z_{l'}-\tilde{z}_{l'})\left(r(i)-\frac{1}{2}a_{ll}(i)\right)\right]\alpha,
\end{align}
using the symmetry of $a^{-1}$. From (\ref{eq:delalpha}), we can see: 
Using the fact that $(i)$ $a^{-1}$ is positive definite (directly follows as $a=\sigma \sigma^T$), and $(ii)$ $z_l=\ln\left(\frac{x_l}{s_l}\right)$, second term has growth $\mathcal{O}(\ln\norm{x}_1)^2$ and is non-negative. This term would be dominating term as $\norm{x}_1\rightarrow \infty$ as the third term has growth of $\mathcal{O}(\ln \norm{x}_1)$ and the first term does not depend on $x_l$.
This implies that for a fixed value of $s$, $i$, and $\sigma$, there exists some large $R$, and constants $C_1$, and $C_2$ (does not depend on $x$ but may depend on $s,i,v$), such that 
\[
\abs{\frac{1}{\alpha}\frac{\partial \alpha}{\partial v}}\leq C_1+C_2\left(\ln\norm{x}_1\right)^2,
\]
for all $x \in (0,\infty)^d \backslash (0,R)^d$.
\end{proof}
\begin{proof}[Proof of Lemma \ref{lemma:delalphadels}]
Recall that $1_l$ denotes the unit vector along $l$th direction and $1_l(l')$ is the $l'$th component of $1_l$. Differentiating (\ref{eq:logalpha}) w.r.t. $s_l$, and using Jacobi's formula for derivative of determinant, we have
\begin{equation*}
\begin{split}
&\frac{1}{\alpha}\frac{\partial \alpha}{\partial s_{l_0}}=0-\frac{1}{2}\times 0-\frac{1}{2}\sum_{l=1}^d \sum_{l'=1}^d \left[(\mathfrak{S}^{-1})_{ll'}\left(\frac{-1}{s_l}\right)(z_{l'}-\tilde{z}_{l'})1_{l_0}(l)+(\mathfrak{S}^{-1})_{ll'}(z_l-\tilde{z}_l)\left(\frac{-1}{s_l}\right)1_{l_0}(l')\right]\\&=\frac{1}{2}\sum_{l' \neq l_0}(\mathfrak{S}^{-1})_{l_0l'}\left(\frac{1}{s_{l_0}}\right)(z_{l'}-\tilde{z}_{l'})+\frac{1}{2}\sum_{l \neq l_0}(\mathfrak{S}^{-1})_{ll_0}\left(\frac{1}{s_{l_0}}\right)(z_l-\tilde{z}_l)+(\mathfrak{S}^{-1})_{l_0l_0}\left(\frac{1}{s_{l_0}}\right)(z_{l_0}-\tilde{z}_{l_0}).
\end{split}
\end{equation*}
Since $\mathfrak{S}$ is symmetric, we can write
\begin{equation*}
\frac{\partial \alpha}{\partial s_{l_0}}=\sum_{l=1}^d (\mathfrak{S}^{-1})_{ll_0}\left(\frac{1}{s_{l_0}}\right)(z_l(s)-\tilde{z}_l(v))\alpha.
\end{equation*}
We have written $z(s)$ and $\tilde{z}(v)$ at the places of $z$ and $\tilde{z}$ to denote their dependency on $s$ and $v$ variables respectively. Therefore from (\ref{eq:alpha}), we have
\[
\frac{\partial \alpha}{\partial s_{l_0}}=\frac{\alpha}{s_{l_0}}\mathcal{O}(\ln{\norm{x}}),\quad \mbox{as} \quad \norm{x}_1 \rightarrow \infty.
\]
\end{proof}
\begin{proof}[Proof of Lemma \ref{lemma:tail}] For notational convenience, let us write $z(s)$ and $\tilde{z}(v)$ to denote their dependency on $s$ and $v$ variables respectively. Furthermore, we denote $(a(i))^{-1}$ by $A$. Moreover, $z(s)-\tilde{z}(v)$ which is a difference of two vectors can be written as $\ln(x)-\left(\ln(s)+\tilde{z}(v)\right)$, where $\ln$ of a vector is the vector of logarithm of components.. Let us write $\ln(s)+\tilde{z}(v)$ as $\xi(s,v)$ for notational convenience. From (\ref{eq:alpha}) and (\ref{eq:logalpha}), we can write
\begin{align*}
    &\ln{\alpha(x,s',i,v')}-\ln{\alpha(x,s,i,v)}\\&=\frac{-d}{2}\ln{\left(\frac{v'}{v}\right)}-\frac{1}{2v'}(\ln{x}-\xi(s',v'))^*A(\ln{x}-\xi(s',v'))+\frac{1}{2v}(\ln{x}-\xi(s,v))^*A(\ln{x}-\xi(s,v)),\\
    &=\frac{-d}{2}\ln{\left(\frac{v'}{v}\right)}-\frac{1}{2v'}\Big[(\ln{x})^*A(\ln{x})-2(\ln{x})^*A\xi(s',v')+\xi(s',v')^*A\xi(s',v')\\
    &-\left(\frac{v'}{v}\right)\Big((\ln{x})^*A(\ln{x})-2(\ln{x})^*A\xi(s,v)+\xi(s,v)^*A\xi(s,v)\Big)\Big],\\
    &=\frac{-d}{2}\ln{\left(\frac{v'}{v}\right)}-\frac{1}{2v'}\left[\left(1-\frac{v'}{v}\right)(\ln{x})^*A(\ln{x})-2(\ln{x})^*A\Big(\xi(s',v')- \left(\frac{v'}{v}\right)\xi(s,v)\Big)
    \right.\\
    &\left.+\xi(s',v')^*A\xi(s',v') - \left(\frac{v'}{v}\right) \xi(s,v)^*A\xi(s,v) \right].
\end{align*}
Clearly, as $\|x\|\to \infty$, the quadratic term $(\ln{x})^*A(\ln{x})$, which appears in the above expression dominates. On the other hand for $v'>v$, the sign of that term is positive. In other words, $\ln{\alpha(x,s',i,v')}-\ln{\alpha(x,s,i,v)}>0$ for large $\|x\|$. To be more precise, for every fixed positive scalar $v'$, $\epsilon$, and vectors $s$, and $s'$, there is a sufficient large $\mathcal{R}$ such that  $\ln{\alpha(x,s',i,v')}-\ln{\alpha(x,s,i,v)}>0$ for all $\| x\| >\mathcal{R}$ and for all $ v\le v'-\epsilon$.
\end{proof}

\begin{proof}[Proof of Lemma \ref{Lemma:theorem2firstpart}]
First we fix the variables $t$, $j$ and $j'$ and hence we ignore their influence on other variables, to be defined in this proof. Since $ \phi( \cdot, \cdot,\cdot)\in V$, for all $s\in(0,\infty)^d$, $\sup_{t,i} \abs{\phi(t,s,i)} \le \norm{\phi}_V (1+\norm{s}_1)$. Let $ \{u_l\}_{l\in\mathbb{N}} $ be a decreasing sequence on $ (0,1) $ such that $ u_l\rightarrow 0 $. Let $ \alpha_l(x):=\alpha(x,s,j',u_l) $. Since $\{ \alpha_l \}_{l\in\mathbb{N}}$ is a family of lognormal density functions with  mean and variance lying on a bounded set, we have the following uniform integrability 
\begin{equation*}
\lim_{R\rightarrow \infty}\sup_l\int_{(0,\infty)^d\backslash (0,R)^d}(1+\norm{x}_1) \alpha(x,s,j',u_l)\,dx=0,
\end{equation*}
for every fixed $s$ and $j'$. Thus, for any $\epsilon>0$, there is $R>0$ such that 
\begin{equation}\label{eq:lemma5_2}
\int_{(0,\infty)^d\backslash (0,R)^d}(1+\norm{x}_1)\alpha(x,s,j',u_l)\,dx<\epsilon \quad \mbox{for all} \quad l\in\mathbb{N}.
\end{equation}
Now consider 
$$\phi_n(t,x,j):=\sum_{i=0}^{2^{2n}-1}\frac{i}{2^n}\textbf{1}_{[\frac{i}{2^n}, \frac{i+1}{2^n})}(\phi(t,x,j)),$$
which is a non-negative increasing sequence converging to $\phi$ uniformly on every compact set. Then, given $\epsilon>0$ and $R$, we can find $N$ such that for all $n\geq N$,
\begin{align}\label{eq:dirac1}
0 & \le \int_{(0,\infty)^d}(\phi(t+u_l,x,j)-\phi_n(t+u_l,x,j))\alpha(x,s,j',u_l) dx,  \nonumber \\
&=\int_{(0,R)^d}(\phi(t+u_l,x,j)-\phi_n(t+u_l,x,j))\alpha(x,s,j',u_l) dx \nonumber\\
 & + \int_{(0,\infty)^d\backslash(0,R)^d  }(\phi(t+u_l,x,j)-\phi_n(t+u_l,x,j))\alpha(x,s,j',u_l) dx, \nonumber\\ 
& \leq \epsilon \int_{(0,R)^d} \alpha(x,s,j',u_l) dx +\int_{(0,\infty)^d\backslash(0,R)^d} \norm{\phi}_V (1+\norm{x}_1)\alpha(x,s,j',u_l) dx,\nonumber\\
 & < (1+ \norm{\phi}_V) \epsilon,
\end{align}
for all $l$, using \eqref{eq:lemma5_2}. 
Also,
\begin{align}\label{sum_i}
\int_{(0,\infty)^d}\phi_n(t+u_l,x,j)\alpha(x,s,j',u_l)dx=&\sum_{i=0}^{2^{2n}-1}\frac{i}{2^n}\int_{A^{(n)}_{l,i}}\alpha(x,s,j',u_l)dx
\end{align}
where $A^{(n)}_{l,i}:=\left\{x\in(0,\infty)^d\mid \phi(t+u_l,x,j) \in[\frac{i}{2^n}, \frac{i+1}{2^n})\right\}$. Now, if 
$i(n,s)$ be such that $\phi(t,s,j) \in[\frac{i(n,s)}{2^n}, \frac{i(n,s)+1}{2^n})$, then for every $s$, due to the continuity of $\phi$ in $t$ variable, there is a sufficiently large $l'$ such that for all $l\ge l'$, $\phi(t+u_l,s,j) \in (\frac{i(n,s)-1}{2^n}, \frac{i(n,s)+1}{2^n})$ as $u_l\downarrow 0$. Therefore, $s$ is in the interior of $A^{(n)}_{l,i(n,s)-1}\cup A^{(n)}_{l,i(n,s)}$ for all $l\ge l'$. Again since variance of $\alpha_l$ goes to zero and mean converges to $s$ as $u_l\downarrow 0$, 
\begin{equation*}
\int_{A^{(n)}_{l,i(n,s)-1}\cup A^{(n)}_{l,i(n,s)}}\alpha(x,s,j',u_l)dx\to 1,
\end{equation*}
because $A^{(n)}_{l,i(n,s)-1}\cup A^{(n)}_{l,i(n,s)}$ contains $s$ in the interior and its Lebesgue measure does not shrink to zero.
Hence the integral of the density function on the complementary domain converges to zero. 
Thus, for each $ n $, using the boundedness of $\phi_n$  and \eqref{sum_i}
\begin{equation*}
\frac{i(n,s)-1}{2^n}\le \lim_{l\rightarrow\infty}\int_{(0,\infty)^d}\phi_n(t+u_l,x,j)\alpha(x,s,j',u_l)dx \le  \frac{i(n,s)}{2^n}.
\end{equation*}
Therefore, from the above inequality
\begin{align*}
&\frac{i(n,s)-1}{2^n}
\leq \lim_{l\rightarrow\infty}\int_{(0,\infty)^d}\phi_n(t+u_l,x,j)\alpha(x,s,j',u_l)\,dx,
\le \lim_{l\rightarrow\infty}\int_{(0,\infty)^d}\phi(t+u_l,x,j)\alpha(x,s,j',u_l)\,dx,\\
&=\lim_{l\rightarrow\infty}\int_{(0,\infty)^d}(\phi(t+u_l,x,j)-\phi_n(t+u_l,x,j)+\phi_n(t+u_l,x,j))\alpha(x,s,j',u_l)\,dx,\\
&\le  (1+ \norm{\phi}_V)\epsilon + \frac{i(n,s)}{2^n},
\end{align*}
using \eqref{eq:dirac1}. For a given $\epsilon>0$ consider $N(>1-\log_2\epsilon)$. Now, from the definition of $i(n,s)$, for all $n\geq N $, $\phi(t,s,j) - \epsilon \le \phi(t,s,j) - \frac{2}{2^n}$  which is less than $\frac{i(n,s)-1}{2^n}$, the left side of above inequality. Finally we note that the right side is less or equal to $(1+ \norm{\phi}_V)\epsilon + \phi(t,s,j)$. Combining these, we get
$$ \phi(t,s,j) - \epsilon \le \lim_{l\rightarrow\infty}\int_{(0,\infty)^d}\phi(t+u_l,x,j)\alpha(x,s,j',u_l)\,dx \le (1+ \norm{\phi}_V)\epsilon + \phi(t,s,j).
$$
The result follows as $\epsilon$ is an arbitrary positive number.
\end{proof}
\begin{proof}[Proof of Identity (\ref{eq:lemma_appendixA})]
From the second order partial derivative of $\alpha$ w.r.t. $s_{l_0}$ and $s_{l'_0}$, we have
\begin{align*}
\frac{\partial}{\partial s_{l'_0}}g_{2}^{l_0}=\frac{1}{s_{l_0}}\frac{\partial}{\partial s_{l'_0}}\sum_{l=1}^d (\mathfrak{S}^{-1})_{ll_0}(z_l-\tilde{z}_l)=\frac{1}{s_{l_0}}\frac{\partial }{\partial s_{l'_0}}\left[(\mathfrak{S}^{-1})_{l'_0l_0}(z_{l'_0}-\tilde{z}_{l'_0})\right]=\frac{-1}{s_{l_0}s_{l'_0}}(\mathfrak{S}^{-1})_{l'_0l_0}-\delta(l_0,l'_0)\frac{1}{s_{l_0}}g_{2}^{l_0},
\end{align*}
where $\delta(l,l')$ is the Kronecker delta function of $l$ and $l'$, i.e., $\delta(l,l')=1$ if and only if $l=l'$ and is zero otherwise. Using the symmetry of $a^{-1}$ and expressions for $g_1$ and $g_2^l$, we simplify the L.H.S. of (\ref{eq:lemma_appendixA}) as follows
\begin{align*}
-&g_1+r(i)\sum_{l_0=1}^d s_{l_0} g_{2}^{l_0}+\frac{1}{2}\sum_{l_0=1}^{d}\sum_{l'_0=1}^d s_{l_0} s_{l'_0}a_{l_0l'_0}\left(\frac{\partial}{\partial s_{l'_0}}g_{2}^{l_0}+g_{2}^{l_0}g_{2}^{l'_0}\right)\\=&-\bigg[ \frac{-d}{2v}-\frac{1}{2}\sum_{l=1}^d\sum_{l'=1}^d\left(\frac{-1}{v^2}a_{ll'}^{-1}(z_l-\tilde{z}_l)(z_{l'}-\tilde{z}_{l'})\right)+\sum_{l=1}^d\sum_{l'=1}^d\frac{a^{-1}_{ll'}}{v}\left(r(i)-\frac{1}{2}a_{ll}\right)\\
& (z_{l'}-\tilde{z}_{l'})\bigg]+r(i)\sum_{l_0=1}^d s_{l_0}\left[\sum_{l=1}^d \frac{a^{-1}_{ll_0}}{v}\left(\frac{1}{s_{l_0}}\right)(z_l-\tilde{z}_l)\right]\\
&+\frac{1}{2}\sum_{l_0=1}^d\sum_{l'_{0}=1}^d s_{l_0}s_{l'_0}a_{l_0l'_0}\Bigg[\left(\frac{-1}{s_{l_0}s_{l'_0}}\frac{a^{-1}_{l'_0l_0}}{v}-\delta(l_0,l'_0)\frac{1}{s_{l_0}}\sum_{l=1}^d\frac{a^{-1}_{ll_0}}{v}\frac{1}{s_{l_0}}(z_l-\tilde{z}_l)\right)\\
&+\left(\sum_{l=1}^d \frac{a^{-1}_{ll_0}}{v}\frac{1}{s_{l_0}}(z_l-\tilde{z}_l) \sum_{l=1}^d \frac{a^{-1}_{ll'_0}}{v}\frac{1}{s_{l'_0}}(z_l-\tilde{z}_l)\right)\Bigg],
\end{align*}
\begin{align*}
=&\frac{d}{2v}+\frac{1}{2}\sum_{l=1}^d\sum_{l'=1}^d\left(\frac{-a^{-1}_{ll'}}{v^2}(z_l-\tilde{z}_l)(z_{l'}-\tilde{z}_{l'})\right)-\sum_{l=1}^d\sum_{l'=1}^d\frac{a^{-1}_{ll'}}{v}\left(r(i)-\frac{1}{2}a_{ll}\right)(z_{l'}-\tilde{z}_{l'})\\
&+r(i)\sum_{l_0=1}^d\left[\sum_{l=1}^d \frac{a^{-1}_{ll_0}}{v}(z_l-\tilde{z}_l)\right]+\frac{1}{2}\sum_{l_0=1}^d \sum_{l'_{0}=1}^d a_{l_0l'_0}\Bigg[\left(-\frac{a^{-1}_{l'_0l_0}}{v}\right)
+\Bigg\{ \sum_{l=1}^d\sum_{l'=1}^d \frac{a^{-1}_{ll_0}a^{-1}_{l'l'_0}}{v^2}\\
&(z_l-\tilde{z}_l) (z_{l'}-\tilde{z}_{l'})\Bigg\}\Bigg]+\frac{1}{2}\sum_{l_0=1}^d \sum_{l'_{0}=1}^d s_{l_0}s_{l'_0}a_{l_0l'_0} \left(-\delta(l_0,l'_0)\frac{1}{s_{l_0}}\sum_{l=1}^d\frac{a^{-1}_{ll_0}}{v}\frac{1}{s_{l_0}}(z_l-\tilde{z}_l)\right),
\end{align*}
\begin{align*}
=&\frac{1}{2}\sum_{l=1}^d\sum_{l'=1}^d\left(\frac{-a^{-1}_{ll'}}{v^2}(z_l-\tilde{z}_l)(z_{l'}-\tilde{z}_{l'})\right)-\frac{r(i)}{v}\sum_{l=1}^d\sum_{l'=1}^da^{-1}_{ll'}(z_{l'}-\tilde{z}_{l'})+\frac{r(i)}{v}\sum_{l=1}^d\sum_{l'=1}^d 
\\&a^{-1}_{ll'}(z_{l}-\tilde{z}_{l})+\sum_{l=1}^d\sum_{l'=1}^d \frac{a_{ll}a^{-1}_{ll'}}{2v}(z_{l'}-\tilde{z}_{l'})+\frac{1}{2}\sum_{l_0=1}^d \sum_{l'_{0}=1}^d a_{l_0l'_0}\Bigg\{ \sum_{l=1}^d\sum_{l'=1}^d \frac{a^{-1}_{ll_0}a^{-1}_{l'l'_0}}{v^2}(z_l-\tilde{z}_l) \\
&(z_{l'}-\tilde{z}_{l'})\Bigg\}-\sum_{l_0=1}^d\sum_{l=1}^d \frac{a_{l_0l_0} a^{-1}_{ll_0}}{2v}(z_{l}-\tilde{z}_{l}),
\end{align*}
\begin{align*}
=&\frac{-1}{2v^2}\sum_{l=1}^d\sum_{l'=1}^d a^{-1}_{ll'}(z_l-\tilde{z}_l)(z_{l'}-\tilde{z}_{l'})+\sum_{l=1}^d\sum_{l'=1}^d \frac{a_{ll} a^{-1}_{ll'}}{2v}(z_{l'}-\tilde{z}_{l'})-\sum_{l_0=1}^d\sum_{l=1}^d \frac{a_{l_0l_0} a^{-1}_{l_0l}}{2v}\\&(z_{l}-\tilde{z}_{l})+\frac{1}{2v^2}\sum_{l'_0=1}^d \sum_{l=1}^d \left[\sum_{l_0=1}^d a_{l'_0 l_0}a^{-1}_{l_0l}\right]\sum_{l'=1}^d a^{-1}_{l'l'_0}(z_l-\tilde{z}_l) (z_{l'}-\tilde{z}_{l'}),
\end{align*}
\begin{align*}
=& \frac{-1}{2v^2}\sum_{l=1}^d\sum_{l'=1}^d a^{-1}_{ll'}(z_l-\tilde{z}_l)(z_{l'}-\tilde{z}_{l'})+\frac{1}{2v^2}\sum_{l'_0=1}^d \sum_{l=1}^d \sum_{l'=1}^d a^{-1}_{l'l'_0}(z_l-\tilde{z}_l) (z_{l'}-\tilde{z}_{l'})\delta(l'_0,l),\\
= & \frac{-1}{2v^2}\sum_{l=1}^d\sum_{l'=1}^d a^{-1}_{ll'}(z_l-\tilde{z}_l)(z_{l'}-\tilde{z}_{l'})+\frac{1}{2v^2}\sum_{l=1}^d \sum_{l'=1}^d a^{-1}_{l'l}(z_l-\tilde{z}_l) (z_{l'}-\tilde{z}_{l'})=0.
\end{align*}
\end{proof}
\begin{proof}[Proof of Lemma \ref{lemma:Stmartingale}]
From the closed form expression of strong solution of SDE, we can write for $t'> t$
\begin{equation*}
\begin{split}
E(\tilde{S}_l(t') | \tilde{\mathcal{F}}_t)&=E\left(\tilde{S}_l(t) e^{(\int_{t}^{t'}(r_u-\frac{1}{2}\hat{\sigma}_l(X_u)^2)du+\int_{t}^{t'} \hat{\sigma}_l(X_u) d\hat{W}^l_u)} | \tilde{\mathcal{F}}_t\right),\\
&=\tilde{S}_l(t) E\left(e^{(\int_{t}^{t'}(r_u-\frac{1}{2}\hat{\sigma}_l(X_u)^2)du+\int_{t}^{t'} \hat{\sigma}_l(X_u) d\hat{W}^l_u)} | \tilde{\mathcal{F}}_t\right),\\
\end{split}
\end{equation*}
where $r_u=r(X_u)$. The conditional distribution of the term inside expectation is log-normal given $\tilde{\mathcal{F}}_t \vee \mathcal{F}^X_{t'}$ with parameters $\left(\int_{t}^{t'} (r_u-\frac{\hat{\sigma}_l(X_u)^2}{2})du , \int_{t}^{t'}\hat{\sigma}_l(X_u)^2 du\right)$. Therefore, we can write
\begin{equation*}
\begin{split}
E(\tilde{S}_l(t') | \tilde{\mathcal{F}}_t)&=\tilde{S}_l(t) E[ e^{\left(\int_{t}^{t'} (r_u-\frac{1}{2} \hat{\sigma}_l(X_u)^2)du+\frac{1}{2}\int_{t}^{t'} \hat{\sigma}_l(X_u)^2 du\right)}\mid\tilde{\mathcal{F}}_{t} ]=\tilde{S}_l(t) E[e^{\int_{t}^{t'} r_u du}\mid X_t],\\
& > \tilde{S}_l(t),
\end{split}
\end{equation*}
using $r_u>0$ for all $u \geq 0$. Therefore, for each $l$, $\tilde{S}_l$ is a sub-martingale.
\end{proof}
\begin{proof}[Proof of Lemma \ref{lemma:Ntmartingale}]
Using the infinitesimal generator of the Markov chain $X$, we can write
\begin{align*}
\phi(t, \tilde{S}(t),X_t)-\phi(t, \tilde{S}(t),X_{t-})= \sum_{j\neq X_{t^-}}\lambda_{X_{t^-},j}\left(\phi(t,\tilde{S}(t),j)-\phi(t,\tilde{S}(t),X_{t^-})\right) dt + dM(t),\end{align*}
for some local martingale $M$. Now,
using It$\hat{o}$'s formula in (\ref{eq:poisson}), we get
\begin{align*}
&dN^\phi_t= -r(X_t)N^\phi_t dt +e^{-\int_{0}^{t}r(X_u)du}\:d \phi(t,\tilde{S}(t),X_t),\\
= & -r(X_t)N^\phi_t dt +e^{-\int_{0}^{t}r(X_u)du} \left[\frac{\partial \phi(t,\tilde{S}(t),X_t)}{\partial t}dt + \sum_{l=1}^d \frac{\partial \phi(t,\tilde{S}(t),X_t)}{\partial s_l}d\tilde{S}_l(t)\right.\\
&+\left.\frac{1}{2}\sum_{l=1}^d\sum_{l'=1}^d a_{ll'}\tilde{S}_l(t)\tilde{S}_{l'}(t)\frac{\partial^2 \phi(t,\tilde{S}(t),X_t)}{\partial s_l \partial s_{l'}}dt+\phi(t, \tilde{S}(t),X_t)-\phi(t, \tilde{S}(t),X_{t-})\right],
\end{align*}
\begin{align*}
&=-r(X_t)N^\phi_t dt +e^{-\int_{0}^{t}r(X_u)du} \Bigg[\Bigg(\frac{\partial \phi}{\partial t} + r(X_t)\sum_{l=1}^d \tilde{S}_l(t) \frac{\partial \phi}{\partial s_l}+\frac{1}{2}\sum_{l=1}^d\sum_{l'=1}^d a_{ll'}\tilde{S}_l(t)\tilde{S}_{l'}(t)\\
&\frac{\partial^2 \phi}{\partial s_l \partial s_{l'}}\Bigg)(t,\tilde{S}(t),X_t)dt+\sum_{j\neq X_{t^-}}\lambda_{X_{t^-},j}\left(\phi(t,\tilde{S}(t),j)-\phi(t,\tilde{S}(t),X_{t^-})\right)dt\\
&+\sum_{l=1}^d\frac{\partial \phi}{\partial s_l}(t,\tilde{S}(t),X_t) \tilde{S}_l(t) \hat{\sigma}_l(X_t)d \hat{W}^l_t+dM(t)\Bigg],
\end{align*}
\begin{align*}
=& e^{-\int_{0}^tr(X_u)du} \Bigg[\Bigg(-r(X_t)\phi(t,\tilde{S}(t),X_t) + \frac{\partial \phi}{\partial t}(t,\tilde{S}(t),X_t)+r(X_t)\sum_{l=1}^d \tilde{S}_l(t)\\
&\frac{\partial \phi}{\partial s_l}(t,\tilde{S}(t),X_t)+\frac{1}{2}\sum_{l=1}^d\sum_{l'=1}^d a_{ll'}\tilde{S}_l(t)\tilde{S}_{l'}(t)\frac{\partial^2 \phi}{\partial s_l \partial s_{l'}}(t,\tilde{S}(t),X_t)+\sum_{j\neq X_{t^-}}\lambda_{X_{t^-},j} dt\\
&\left(\phi(t,\tilde{S}(t),j)-\phi(t,\tilde{S}(t),X_{t^-})\right)\Bigg)+\frac{\partial \phi}{\partial s_l} (t,\tilde{S}(t),X_t)\tilde{S}_l(t) \hat{\sigma}_l(X_t)d \hat{W}^l_t + dM(t)\Bigg].
\end{align*}
\noindent It is clear that the coefficient of $dt$ is zero from (\ref{eq:main_ivp})-(\ref{eq:operator}). Thus $N^\phi$ is a local martingale. On the other hand since $\phi\in V$,  $|N^\phi_t| \leq \|\phi\|_V +  \|\phi\|_V \sum_l \tilde{S}_l(t)$ for each $t$. Therefore, $E\left(\sup_{s \leq t}|N^\phi_s|\right)< \infty$  using \eqref{DMI}. Thus, $N^\phi$ is a martingale, follows from \cite[Theorem $51$, Chapter $1$, pp. $38$]{Protter90}.
\end{proof}
\section{Derivation of IE \eqref{eq:integral_equation}}\label{sec:appendixB}
\par Let $n(t)$ denote the number of transitions during $(0,t]$ and $T_n$ denote the $n^{th}$ transition time instant. We rewrite the right side of (\ref{eq:option_expectation}) by conditioning on the next transition time $T_{n(t) +1}$
\begin{equation}\label{eq:tower_phi}
\phi(t, \tilde{S}(t), X_t)=E\left[E\left(e^{-\int_{t}^T r(X_u)du}\mathcal{K}(\tilde{S}(T))|\tilde{S}(t),X_t,T_{n(t)+1}\right)|\tilde{S}(t), X_t\right].
\end{equation}
\begin{lemma}\label{lemma:expo_dis}
$T_{n(t)+1}-t$ is exponentially distributed random variable given ${\mathcal{F}}^{X}_t$.
\end{lemma}
\begin{proof}
The conditional cumulative distribution function (CDF) of $T_{n(t)+1}$ given $X_t=i$ is
\begin{equation*}
F_i(v):=P(T_{n(t)+1} \leq v | X_t=i)= P(T_{n(t)+1}-T_{n(t)} \leq v-T_{n(t)}|X_{t}=i),
\end{equation*}
for $v >t$ and zero for $v\le t$ as $T_{n(t)+1}>t$ almost surely. For the same reason, $T_{n(t)+1}-T_{n(t)} \geq t-T_{n(t)}$ almost surely. Hence, we can write
\begin{align*}
F_i(v)&=P(T_{n(t)+1}-T_{n(t)} \leq v-T_{n(t)}|X_t=i),\\
&= P(T_{n(t)+1}-T_{n(t)}\leq v-T_{n(t)}| X_t=i, T_{n(t)+1}-T_{n(t)}\geq t-T_{n(t)}).
\end{align*}
Using Bay's theorem the above is equal to 
\begin{align*}
&\frac{P(t-T_{n(t)}\leq T_{n(t)+1}-T_{n(t)}\leq v-T_{n(t)}| X_t=i)}{P(T_{n(t)+1}-T_{n(t)}\geq t-T_{n(t)}|X_t=i)}.
\end{align*}
By additional conditioning by $T_{n(t)}$, the above is rewritten as 
\begin{align*}
\frac{E[P(t-T_{n(t)}\leq T_{n(t)+1}-T_{n(t)}\leq v-T_{n(t)}| X_t,T_{n(t)})|X_t=i]}{E[P(T_{n(t)+1}-T_{n(t)}\geq t-T_{n(t)}|X_t,T_{n(t)})|X_t=i]}.
\end{align*}
Using the exponential distribution of inter transition time of the Markov chain $X$, the above is
\begin{align*}
\frac{E[(1-e^{-\lambda_{X_t}(v-T_{n(t)})})-(1-e^{-\lambda_{X_t} (t-T_{n(t)})})|X_t=i]}{E[e^{-\lambda_{X_t} (t-T_{n(t)})}|X_t=i]}
&=\frac{E[e^{-\lambda_{X_t}(t-T_{n(t)})} (1-e^{-\lambda_{X_t}(v-t)})|X_t=i]}{E[e^{-\lambda_{X_t} (t-T_{n(t)})}|X_t=i]},\\
&=1-e^{-\lambda_i(v-t)}.
\end{align*}
Therefore, by substituting $v=v'+t$, we get the probability density function of $T_{n(t)+1}-t$ as $\frac{d}{dv'}F_i(v'+t)=\lambda_i e^{-\lambda_i v'}$ on $v'>0$ and result follows.
\end{proof}
\begin{proposition}\label{proposition:uniqueness}
Let $\phi$ be a classical solution of (\ref{eq:main_ivp})-(\ref{eq:operator}), then $\phi$ also solves integral equation (\ref{eq:integral_equation}).
\end{proposition}
\begin{proof}
First we note that $T_{n(t)+1}-t$ is conditionally independent to $\tilde{S}(t)$ given  $X_t$. 
Let us consider (\ref{eq:tower_phi}) and using formula for expectation and Lemma \ref{lemma:expo_dis}, we get
\begin{align}\label{eq:option_expec1}
\phi(t,\tilde{S}(t),X_t)=&\int_{0}^\infty E\left[e^{-\int_t^T r(X_u) du}\mathcal{K}(\tilde{S}(T))|\tilde{S}(t),X_t,T_{n(t)+1}-t=v'\right]\lambda_{X_t} e^{-\lambda_{X_t}v'}dv',\nonumber\\
=&\int_{0}^{T-t} E\left[e^{-\int_t^T r(X_u) du}\mathcal{K}(\tilde{S}(T))|\tilde{S}(t),X_t,T_{n(t)+1}-t=v'\right]\lambda_{X_t} e^{-\lambda_{X_t}v'}dv'\nonumber\\
& + \int_{T-t}^\infty E\left[e^{-\int_t^T r(X_u) du}\mathcal{K}(\tilde{S}(T))|\tilde{S}(t),X_t,T_{n(t)+1}-t=v'\right]\lambda_{X_t} e^{-\lambda_{X_t}v'}dv'.
\end{align}
Note that $ T_{n(t)+1}-t= v'\in [T-t, \infty)$ implies $\{T_{n(t)+1}\geq T\}$. Again, under $T_{n(t)+1}\geq T$, no transition takes place during $[t,T]$. Thus $X_T$ is identical to $X_t$. Hence, the conditional distribution of each component of $\tilde{S}(T)$ given the $\sigma$ algebra generated by  $\{T_{n(t)+1}>T\}$ and  $\tilde{\mathcal{F}}_t$ is log normal. 
That is given this $\sigma$ algebra, with $X(t)=i$ the conditional joint distribution of $\{\tilde{S}_{l}(T)/\tilde{S}_{l}(t): l= 1,2,\ldots,d\}$ is identical to the joint distribution of $\{Y_{l}(T)/Y_{l}(t): l= 1,2,\ldots,d \}$, which corresponds to the B-S-M model having constant parameters $r(i)$ and $\sigma(i)$. Thus the conditional expectation $E\left[e^{-r(i)(T-t)}\mathcal{K}(\tilde{S}(T))|\tilde{S}(t)=s,X_u=i \:\forall\: u \in [t,T]\right]$ is identical to the R.H.S. of (\ref{eq:eta}). Therefore, (\ref{eq:option_expec1}) can be written as
\begin{align}\label{eq:prepo1_phi}
&\phi(t,\tilde{S}(t),X_t)=\int_{0}^{T-t} E\left[e^{-\int_t^T r(X_u) du}\mathcal{K}(\tilde{S}(T))|\tilde{S}(t),X_t,T_{n(t)+1}-t=v'\right]\lambda_{X_t} e^{-\lambda_{X_t}v'}dv'\nonumber\\
&+\eta_{X_t}(t,s)\int_{T-t}^{\infty}\lambda_{X_t} e^{-\lambda_{X_t} v'}dv',\\
&=\int_{0}^{T-t} E\left[e^{-\int_t^T r(X_u) du}\mathcal{K}(\tilde{S}(T))|\tilde{S}(t),X_t,T_{n(t)+1}-t=v'\right]\lambda_{X_t} e^{-\lambda_{X_t}v'}dv'+\eta_{X_t}(t,s)e^{-\lambda_{X_t} (T-t)}\nonumber.
\end{align}
Further, given $T_{n(t)+1}= t+ v'<T$, it is clear that during $[T_{n(t)},t+v')$, $X$ has no transition. Moreover, at time $t+v'$, $X$ transits to another state $j$ with conditional probability $P_{X_t,j}$ given $X_t$. Given $\tilde{S}_l(t)=s_l, X_t=i$, and $T_{n(t)+1}=t+v'$, we know that
\[
\tilde{S}_l(t+v')=s_l\exp\left[(r(i)-\frac{1}{2}a_{ll'}(i))v'+\sum_{j=1}^d\sigma_{l,j}(i)(W_{t+v'}-W_t)\right],
\]
is log-normal for each $l=1,2, \ldots,d$. The joint conditional PDF of $\tilde{S}(t+v')$ is given by $\alpha$ as in (\ref{eq:alpha}). In (\ref{eq:prepo1_phi}) by fixing $\tilde{S}(t)=s$, $X_t=i$, and using
additional conditioning w.r.t. $\tilde{\mathcal{F}}_{t+v'}$, we get
\begin{align}
\phi(t,s,i)=&\int_{0}^{T-t} \! \!E\left[E\Big(e^{-\left(r(i)v'+\int_{t+v'}^{T}r(X_u) du\right)}\mathcal{K}(\tilde{S}(T))| \tilde{S}(t+v'), X_{t+v'},\tilde{S}(t),X_t, \right.\nonumber\\
&  T_{n(t)+1}-t=v'\Big) \left. \lvert \: \tilde{S}(t)=s,X_t=i, T_{n(t)+1}-t=v'\right]\lambda_i e^{-\lambda_iv'}dv'+e^{-\lambda_i(T-t)}\eta_i(t,s),\nonumber\\
=&\int_{0}^{T-t}\sum_{j=1}^k p_{ij}\Big[\int_{(0,\infty)^d}E\left(e^{-\left(r(i)v'+\int_{t+v'}^{T}r(X_u) du\right)}\mathcal{K}(\tilde{S}(T))|\tilde{S}(t+v')=x, X_{t+v'}=j\right)\nonumber\\
&\alpha(x,s,i,v') dx \Big] \lambda_i e^{-\lambda_iv'}dv'+e^{-\lambda_i(T-t)}\eta_i(t,s),\nonumber\\
=&\int_{0}^{T-t}e^{-(\lambda_i+r(i))v'}\sum_{j=1}^k \lambda_{ij}\int_{(0,\infty)^d}E\Big(e^{-\int_{t+v'}^{T}r(X_u) du}\mathcal{K}(\tilde{S}(T))|\tilde{S}(t+v')=x, \nonumber\\
& X_{t+v'}=j\Big) \alpha(x,s,i,v') dx dv'+e^{-\lambda_i(T-t)}\eta_i(t,s).\nonumber
\end{align}
Using (\ref{eq:option_expectation}), we get
\begin{align*}
\phi(t,s,i)=e^{-\lambda_i(T-t)}\eta_i(t,s)+\int_{0}^{T-t}e^{-(\lambda_i+r_i)v'}\sum_{j=1}^k \lambda_{ij}\int_{(0,\infty)^d}\phi(t+v',x,j)
\alpha(x,s,i,v') dx dv'.
\end{align*}
Thus $\phi$ also solves the integral equation (\ref{eq:integral_equation}). This completes the proof.
\end{proof} 
\noindent We recall that uniqueness of (\ref{eq:main_ivp})-(\ref{eq:operator}) has been established in Theorem \ref{theorem:uniqueness}. Here, we present an alternative argument for the same. Let us assume that $\phi_1$ and $\phi_2$ are two classical solutions of (\ref{eq:main_ivp})-(\ref{eq:operator}) in $V$. Then using Proposition \ref{proposition:uniqueness}, we know that both also solve IE (\ref{eq:integral_equation}). But from Theorem \ref{theorem:IE_solution}, there is only one such solution in $V$. Hence $\phi_1=\phi_2$.

\section*{Acknowledgments}
The corresponding author acknowledge the support from IISER Pune and IIIT Naya Raipur for providing the support to conduct some part of this research work. 

\bibliographystyle{plain}
\bibliography{main_MMGBM}

\end{document}